\documentclass[11pt]{article}

\usepackage{etoolbox}

\newbool{notedraft}%
\booltrue{notedraft}%

\usepackage[margin=1in]{geometry}

\ifbool{notedraft}{
  \newcommand{\whendraft}[1]{#1}%
}{
  \newcommand{\whendraft}[1]{}%
}

\usepackage[hyphens]{url}

\usepackage[style=alphabetic,natbib=true,maxalphanames=10,maxnames=10,sortcites=true,doi=false,url=false]{biblatex} %

\usepackage{amsmath}
\usepackage[citecolor=magenta,anchorcolor=blue,colorlinks]{hyperref}%
\usepackage{cleveref}

\usepackage{microtype}          %
\usepackage[T1]{fontenc}

\usepackage{lmodern}
\usepackage[full]{textcomp}     %

\usepackage{article} %
\usepackage{math} %
\usepackage{algo} %
\usepackage{wrapfig} %
\usepackage{placeins} %

\usepackage{appendix}           %

\usepackage{graphicx} %
\usepackage{layout} %
\usepackage{setspace} %
\usepackage{mathtools} %
\usepackage{grffile} %
\usepackage{pdfpages} %
\usepackage{mdframed} %
\usepackage{bm} %
\usepackage{needspace} %
\usepackage{pbox} %

\usepackage{titlesec}

\usepackage{truncate}

\usepackage{marginnote}%

\whendraft{
  \usepackage{marginfix}
}

\NewDocumentEnvironment{Note}{+b}{%
  \whendraft{%
    \ignorespaces%
    \marginpar{\footnotesize\noindent#1}%
    \ignorespaces%
  }
}%

\NewDocumentCommand{\undernote}{s O{blue} m m}{
  \IfBooleanT{#1}{\smash}%
  {\color{#2} %
    \underbrace{\normalcolor%
      #4}_{\mathclap{\text{#3}}} %
  }%
  \IfBooleanT{#1}{\vphantom{#4}}
}

\newcommand{\defterm}[1]{{\boldmath\normalfont \bfseries #1}}%
\renewcommand{\defterm}{\emph}%

\makeatletter
\g@addto@macro\bfseries{\boldmath}
\makeatother

\NewDocumentEnvironment{rightfigmp}{O{1em} m m}{
  \noindent%
  \begin{minipage}[t]{\linewidth}%
    \noindent%
    \begin{minipage}[t]{\linewidth - #2 - #1}%
    }{
    \end{minipage}%
    \hfill%
    \begin{minipage}[t]{#2}%
      \includegraphics[width=#2,valign=t]{figures/#3}%
    \end{minipage}
  \end{minipage}
}

\titlespacing*{\paragraph}{%
  0pt}{%              left margin
  {\medskipamount}}{% space before (vertical)
  1em}%               space after (horizontal)

\providecommand{\wrt}{with respect to\xspace}%

\bibliography{references,density}

\begin{document}

\newcommand{\dsg}{\textsc{DSG}\xspace}    %
\newcommand{\arboricity}{\fparnew{\operatorname{arb}}}
\providecommand{\y}{\fparnew{y}}
\providecommand{\bary}{\fparnew{\bar{y}}}
\providecommand{\cut}{\fparnew{\partial}}%

\providecommand{\incut}{\fparnew{\delta^-}}
\providecommand{\outcut}{\fparnew{\delta^+}}
\providecommand{\indegree}{\fparnew{d^-}}

\providecommand{\lpopt}{\opt_{\operatorname{LP}}}%
\newcommand{\edge}[2]{\setof{#1,#2}}%

% \NewDocumentCommand{\indegreeT}{m}{
%   \min{\indegree{#1}, T}%
% }

\NewDocumentCommand{\indegreeT}{m}{
  \indegree{#1}_T%
}

\providecommand{\outdegree}{\fparnew{d^+}}
\providecommand{\inneighbors}{\fparnew{N^-}}
\providecommand{\outneighbors}{\fparnew{N^+}}

\providecommand{\mini}{i_0}%
\providecommand{\alphamore}{\parof{1 + \alpha}}%
\providecommand{\level}{\fparnew{\operatorname{level}}}%
\providecommand{\lev}{\level}%
\providecommand{\lab}{\fparnew{\operatorname{label}}}%
\renewcommand{\lab}{\fparnew{\varphi}}%
\providecommand{\alphaless}{\parof{1 - \alpha}}%
\providecommand{\therank}{r}%
\providecommand{\thesize}{p}%

\title{$(1-\eps)$-approximate fully dynamic densest subgraph: \\ linear
  space and faster update time}

 \author{ Chandra
   Chekuri\thanks{Dept.\ of Computer Science, Univ.\ of Illinois,
     Urbana-Champaign, Urbana, IL 61801. {\tt
       chekuri@illinois.edu}. Supported in part by NSF grants
     CCF-1910149 and CCF-1907937.}  \and Kent Quanrud\thanks{Dept.\ of Computer
     Science, Purdue University, West Lafayette, IN 47907. {\tt
       krq@purdue.edu}. Supported in part by NSF grant CCF-2129816.}  }

\date{}

\maketitle

\begin{abstract}
  We consider the problem of maintaining a $(1-\epsilon)$-approximation to
  the densest subgraph (\dsg) in an undirected multigraph as it
  undergoes edge insertions and deletions (the fully dynamic setting).
  Sawlani and Wang \cite{SW20} developed a data structure that, for
  any given $\epsilon > 0$, maintains a $(1-\epsilon)$-approximation with
  $O(\log^4 n/\epsilon^6)$ worst-case update time for edge operations, and
  $O(1)$ query time for reporting the density value.  Their data
  structure was the first to achieve near-optimal approximation, and
  improved previous work that maintained a $(1/4 - \epsilon)$
  approximation in amortized polylogarithmic update time
  \cite{BHNT15}.  In this paper we develop a data structure for
  $(1-\epsilon)$-approximate \dsg that improves the one from \cite{SW20}
  in two aspects. First, the data structure uses linear space
  improving the space bound in \cite{SW20} by a logarithmic
  factor. Second, the data structure maintains a
  $(1-\epsilon)$-approximation in amortized $O(\log^2 n/\epsilon^4)$ time per
  update while simultaneously guaranteeing that the worst case update
  time is $O(\log^3 n \log \log n/\epsilon^6)$. We believe that the space
  and update time improvements are valuable for current large scale
  graph data sets.  The data structure extends in a natural fashion to
  hypergraphs and yields improvements in space and update times over
  recent work \cite{BeraBCG22} that builds upon \cite{SW20}.
\end{abstract}

\section{Introduction}
The \emph{densest subgraph problem} (\dsg) is the following. Given an
undirected (multi)graph $G=(V,E)$ find
$\max_{S \subseteq V} \frac{|E(S)|}{|S|}$ where $E(S)$ is the set of
edges with both endpoints in $S$. We use $\rho_G(S)$ to denote
$\frac{|E(S)|}{|S|}$, the density of $S$ in the graph $G$. \dsg is a
very well-studied problem with both practical and theoretical
appeal. On the practical side it is a core problem in network analysis
and graph mining to find clusters and communities. In addition to
being directly relevant, it is also a canonical problem in the field
of dense subgraph discovery \cite{lrja-10,gt-15,tc-21}.  On the
theoretical side \dsg is polynomial-time solvable problem and has
several important connections to network flow, arboricity, matchings,
submodular, and supermodular optimization
\cite{goldberg,pq-82,fujishige,fujishige09,BHNT15,SW20,ChekuriQT22}. There
have been several recent works on algorithmic aspects of \dsg and its
variants. A particular emphasis has been on fast and scalable
algorithms due to the large scale graph data driving many of the
applications. For this reason, even though there are polynomial-time
exact algorithms for \dsg via network flow \cite{goldberg,pq-82} and
submodular function minimization, the recent focus has been on
near-linear time constant factor and $(1-\eps)$-approximation
algorithms and heuristics
\cite{Charikar00,saha-directed-dsg,bgm-14,bsw-19,frankwolfe,flowless,ChekuriQT22}.\footnote{Maximum
  flow has seen a spate of breakthroughs in the last decade or more
  and a near-linear time algorithm was announced just a few months ago
  \cite{maxflow22}.  Despite their theoretical importance, these new
  algorithms are far from practical at this point in time.} Another
reason for the focus on approximation is that exact efficient
algorithms are unlikely in other models of computing such as streaming
\cite{bkv-12,andrewMc-math-foundations,BHNT15}, parallel and
distributed \cite{bkv-12,bgm-14,glm-19}, and dynamic
\cite{AlessandroWWW,BHNT15,SW20}, which are of much importance.
Constrained versions of \dsg such as the densest $k$-subgraph problem
\cite{fpk-01,k-06,bccfv-10,m-17} are also of much interest and very
well-studied, especially for their theoretical importance. In this
paper we focus on the unconstrained \dsg problem.

In this paper we are interested in the \emph{dynamic} setting for \dsg
where the underlying graph $G$ undergoes edge insertions and
deletions.  The goal is to maintain an approximation to the densest
subgraph.  This is the so-called \emph{fully} dynamic setting.%
\footnote{We focus on edge insertions and deletions assuming that the
  number of vertices is fixed and known. One can handle vertex
  additions and deletions via edge updates in standard ways. We
  believe that edge updates are more interesting from both a practical
  and theoretical point of view for \dsg.} The two query operations
to the data structure are: (i) report the (approximate) value of the
densest subgraph and (ii) output an (approximately optimal) densest
subgraph. Note that the densest subgraph can be very large, and hence
the typical focus is to have a fast query time for the value, and to
be able to output the densest subgraph in time proportional to its
size. In many applications of interest the underlying graph is large
and changes frequently \cite{Sahuetal17}. In addition to the practical
interest, dynamic (graph) algorithms have seen a surge of interest and
many exciting new results in recent years, and they have led to
numerous breakthroughs in obtaining faster algorithms for a number of
fundamental problems in a variety of models --- we refer the reader to
some surveys \cite{Henzinger18,HanauerHS21} rather than give pointers
to the very large literature.  For dynamic \dsg, Sawlani and Wang
\cite{SW20}, in an elegant recent work, developed a data structure
that for any given $\eps > 0$, maintains a $(1-\eps)$-approximation to
the value of the optimum density in \emph{worst-case} update time of
$O((\log^4 n)/\eps^6)$ using $O(m \log m)$ space, where $n$ and $m$
are the number of vertices and edges respectively. Their data
structure reports the value of the density in $O(1)$ time.  Their work
is the first to maintain an arbitrarily good approximation for the
optimum density, improving previous works, and in particular the work
of Bhattacharya et al. \cite{BHNT15} that maintained a
$(1/4 -\eps)$-approximation with a polylogarithmic \emph{amortized}
update time.

\subsection{Motivation and contribution}
Our goal is to obtain improved data structures for dynamic \dsg.
Given the large scale graphs that are common today, a practical
concern is the space usage. For instance, if a graph has a million
edges then $\log m$ (to base 2) is about $20$. A data structure that
uses $O(m \log m)$ space instead of linear space, like the one in
\cite{SW20}, may not be able to fit data in main memory even though
the constant may appear small in theoretical analysis. Second, it is
helpful to obtain improved update times since it will allow for more
accurate estimates under some given budget on the update time. In this
paper we develop fully dynamic data structures for
$(1-\eps)$-approximate fully dynamic \dsg that have the following
features:
\begin{itemize}
\item The space used is linear in the size of the graph.
\item The query time to answer the value of the densest subgraph is $O(1)$.
\item The data structure reports a value $\lambda$ such that $\lambda
  \ge (1-\eps)\opt - O(\ln n/\eps)$ in \emph{amortized} $O(\log
  n/\eps^2)$ time for edge updates. Here $\opt$ is the optimum density.
  A $(1-\eps)$ true approximation can be maintained in amortized $O(\log^2 n/\eps^4)$ time.
\item The data structure can be extended to have a \emph{worst-case}
  update time of $O(\log^3 n\log \log n/\eps^6)$ while maintaining the
  amortized bound of $O(\log^2 n/\eps^4)$ for edge operations.
\end{itemize}
Thus we are able to improve upon the data structure of \cite{SW20} in
terms of space, and also obtain faster worst-case and amortized update
times.  In addition to obtaining improved bounds, our data structure
is quite simple and the analysis is self-contained. Like previous data
structures for density, ours is also based on maintaining graph
orientations to minimize arboricity. Along the way we obtain some new
tradeoffs for maintaining arboricity that are of independent
interest. We outline these in the next subsection.

\paragraph{Extensions:} Our data structure for maintaining
$(1-\eps)$-approximation for \dsg extends in a natural and relatively
easy fashion to hypergraphs. We obtain a linear space data structure
and the running time for a rank $r$ hypergraph is an $O(r^2)$ factor
worse than it is for graphs. Recent work of Bera et
al. \cite{BeraBCG22} has extended the ideas in \cite{SW20} to
hypergraphs; our results improve upon theirs in a similar fashion as
ours improve upon \cite{SW20} for graphs in terms of space and update
times.  Bera et al. \cite{BeraBCG22} showed that one can handle
arbitrary edge-weights in graphs and hypergraphs via a sparsification
technique \cite{mpptx-15} and a black-box reduction to a data
structure for the unweighted case. Their reduction is randomized and
assumes an oblivious adversary. Moreover, the reduction requires
guessing the optimum density and maintaining a logarithmic number of
parallel copies of the unweighted data structure. This increases the
space and update times by poly-logarithmic factors when compared to
the unweighted case. We can employ their reduction and plug in our
data structure for the unweighted case.

\subsection{Technical overview in the context of related work}
The optimum density of a given graph $G$ is closely related to the
graph theoretic notion of \emph{arboricity} that we now define. Given
an undirected graph $G=(V,E)$, an orientation of $G$ is a directed
graph $D=(V,A)$ that is obtained from $G$ by orienting each edge $\edge{u}{v}
\in E$ either as the arc $(u,v)$ or as the arc $(v,u)$.  Given an
orientation $D$ of $G$, let $\arboricity(G,D)$ be the maximum
in-degree among vertices in $D$.  The arboricity of $G$, denoted by
$\arboricity(G)$, is the minimum $\arboricity(G,D)$ over all
orientations $D$ of $G$.  A well-known theorem of Nash-Williams
\cite{NashW64} implies that $\arboricity(G)$ is equal to the minimum
number of forests that are needed to cover the edge set $E$. Further,
$\arboricity(G)$ can be computed in polynomial time. One can show that
$\rho(G)$, the optimum density of a graph $G$, corresponds to the
\emph{fractional} arboricity of $G$; by fractional one means that an
edge $\edge{u}{v} \in E$ is now allowed to be fractionally oriented both as
$(u,v)$ and $(v,u)$ (with the sum of the non-negative fractions summing to one) and
the goal is to minimize the maximum fractional in-degree of the nodes.
A fractional orientation can be viewed as a solution to an exact LP
relaxation for \dsg suggested by Charikar \cite{Charikar00} (this has
been noted in several papers including \cite{bgm-14,flowless,SW20}). It can be shown
that the fractional and integral arboricity differ by at most $1$.

Independent of the connection to \dsg, dynamic maintainance of the
arboricity of a graph has received attention in the data structures
community since it has connections to the problem of maintaining fast
adjacency queries in low arboricity graphs (such as planar graphs)
\cite{KannanNR92}. Several papers, starting with the work of Brodal
and Fagerburg \cite{BrodalF99}, developed dynamic algorithms (and
analysis) for maintaining approximate arboricity. The initial papers
had amortized complexity bounds assuming that the arboricity was
guaranteed to be upper bounded by a given bound $\alpha$
\cite{BrodalF99,Kowalik07}.  Kopeliwitz et al \cite{KKPS14} were the
first to obtain a data structure that had polylogarithmic worst-case
update time.  Specifically, their algorithm maintained an orientation
such that the maximum in-degree of the orientation is
$(1+\eps)\arboricity(G) + O(\log n/\eps)$. However the update time
depended on the arboricity (in the application of interest, arboricity
was small and this was not a limitation). In these papers the focus
was on maintaining a somewhat loose approximation to the arboricity
(such as a constant factor with the precise constant left unspecified)
as the approximation translated into running time rather than the
quality of a solution. Subsequently, motivated by application to \dsg,
Bhattacharya et al. \cite{BHNT15} developed a data structure that
maintained constant factor approximation to the arboricity in
poly-logarithmic amortized update time and linear space (see also
\cite{HenzingerNW20}). Their update time did not depend on the
arboricity. Using this, \cite{BHNT15} showed that a
$(1/4 -\eps)$-approximate densest subgraph can be maintained in
amortized poly-logarithmic update time. We note that there is no data
structure so far that can maintain a constant factor approximation to
the arboricity with a worst-case poly-logarithmic update time.

To obtain a $(1-\eps)$-approximation to density with worst-case update
time, Sawlani and Wang \cite{SW20} use two important ideas. First,
they exploited the fact that density corresponds to fractional
arboricity. This additional flexibility allowed them to make copies of
edges and assume that that the arboricity is sufficiently large.  They
are then able to use the above mentioned algorithm of Kopeliwitz et al
\cite{KKPS14} which maintain $(1+\eps)$-approximation to arboricity
but has an additive error which can be absorbed when arboricity is
large. To overcome the issue that the update time of the data
structure in \cite{KKPS14} has a dependence on the arboricity,
\cite{SW20} uses a simple form of scaling by ``guessing/estimating''
the optimum density. However, the price to make this idea work is that
\cite{SW20} need to maintain $\Omega(\log m)$ copies of the data
structure from \cite{KKPS14}, one for each potential value of the
density (within a factor of $2$).  They need to dynamically adjust all
these data structures simultaneously, and need additional
modifications to overcome the running time dependency in \cite{KKPS14}
on the arboricity, which could translate to bad running times for the
copies of the data structure corresponding to scales smaller than the
optimum density. This is the reason for an additional logarithmic
factor in the space and some complexity in the implementation of the
data structure.

In this paper we build on previous ideas but take a somewhat different
approach. The algorithms for maintaining a low arboricity orientation
(here and in prior work) use certain invariants that try to
balance % the difference in
the in-degrees between adjacent vertices --- the high-level goal is to
orient each edge towards the lower degree vertex.  They flip (i.e.,
reorient) edges to maintain the invariants as edges are inserted and
deleted. However, flipping one edge can lead to flipping an adjacent
edge and eventually cause a cascading sequence of flips, so we require
a careful analysis to argue about the update time and the quality of
the orientation.  Our first idea is that one can maintain an
orientation of the graph such that maximum in-degree is at most
$(1+\eps)\arboricity(G) + O(\log n/\eps)$ in amortized
$O(\log n/\eps)$ update time. To obtain a worst-case guarantee we
alter the update algorithm in two important, but relatively simple (in
retrospect) ways. Note that, unlike the work in \cite{KKPS14}, the
update time does \emph{not} depend on the arboricity while providing
the same guarantee.  As far as we are aware, there was no previous
dynamic algorithms that maintains an orientation with maximum
in-degree arbitrarily close to arboricity in worst-case (or even
amortized) polylogarithmic update time in the large arboricity regime.
Once we have the above guarantee, we can use the idea of duplicating
edges (which can be done implicitly and does not add to the space) to
maintain a $(1-\eps)$-approximation for fractional orientations which
corresponds to density.

Our data structures maintain a local optimality invariant on the
orientation that differs in a simple but crucial way from that of
\cite{KKPS14}.  In a certain sense, we take a first-principles
approach to maintaining orientations with both additive and
multiplicative slack, given that density maintenance can absorb the
additive slack via edge duplication. We believe that this leads to a
clean and improved data structure. The transparency of the analysis
also allowed us to improve the worst-case update by exploiting the
different behavior of the basic data structure when the arboricity is
large and when it is small.

\subsection{Other related work}
The first polynomial-time algorithm for \dsg was via a reduction to
network flow \cite{goldberg,pq-82}; the decision problem of whether
$\rho(G) \ge \lambda$ can be solved via $s$-$t$ maxflow in an
auxiliary graph which has $|E|$ edges and $|E|$ vertices.  Combining
this binary search over $\lambda$ yields an algorithm to find the
optimum density.  This leads to a near-linear time algorithm via the
current fastest algorithm for $s$-$t$ maximum flow when the edge and
vertex weights are polynomially bounded \cite{maxflow22}.  One can
also derive a polynomial-time algorithm via reduction to submodular
function minimization via the following observation: for any graph
$G=(V,E)$, the set function $f:2^V \rightarrow \mathbb{Z}_+$ defined
by $f(S) = |E(S)|$ is supermodular.\footnote{A real-value set function
  $f:2^V \rightarrow \mathbb{R}$ is submodular iff
  $f(A) + f(B) \ge f(A \cup B) + f(A \cap B)$ for all
  $A, B \subseteq V$. A set function is supermodular iff $-f$ is
  submodular.} The supermodularity perspective allows one to handle
generalizations of \dsg and other problems --- we refer the reader to
\cite{fujishige09,veldt-dsg-p,ChekuriQT22}.  Charikar, in an
influential work \cite{Charikar00}, showed that a simple greedy
algorithm \cite{aitt-00} yields a $\frac12$-approximation for \dsg,
and he also described an LP relaxation that is exact for DSG. The dual
of this LP can be viewed as finding the minimum degree fractional
orientation of the given graph.  We discuss more details in
\refsection{prelim}. The LP relaxation has led to several fast
$(1-\eps)$-approximation algorithms via mathematical programming and
flow tehniques \cite{bgm-14,flowless,bsw-19,ChekuriQT22}.

As we remarked, algorithms for \dsg have been explored in the last few
years in streaming, mapreduce, parallel and distributed, and dynamic
settings, starting with the work of \cite{bkv-12}. In terms of dynamic
data structures, around the same time as \cite{BHNT15}, Epasto et al
\cite{AlessandroWWW} described a data structure that maintained a
$(1/2-\eps)$-approximation with insertions only in amortized
$O(\frac{1}{\eps^2} \log^2 n)$ time; they also generalized their
result to handle random deletions with slightly worse update time. Hu,
Wu and Chan \cite{HuWC17} were the first to consider dynamic densest
subhypergraph. Their results are parametrized by the rank $r$ of the
hypergraph. They maintained a $\frac{1-\eps}{r}$-approximation in the
insertions only case, and a $\frac{1-\eps}{r^2}$-approximation in the
fully dynamic case. Their update time is amortized
$\text{poly}(\frac{r}{\eps}\log n)$.  Recently Bera et
al. \cite{BeraBCG22}, building upon \cite{SW20}, showed that one can
maintain a $(1-\eps)$-approximation for densest subhypergraph in the
fully dynamic setting with a worst-case update time of
$\text{poly}(\frac{r}{\eps}\log n)$.  As we mentioned, they also
extended the algorithm to the weighted case, via randomization,
against oblivious adversaries.

Kannan and Vinay introduced a directed graph version of \dsg
\cite{kv-99}. Charikar \cite{Charikar00} showed that it can be solved
exactly via a reduction to polynomial number of \dsg instances with
vertex weights.  Improvements in the running time were made by Saha
and Khuller \cite{ks-09}. Sawlani and Wang \cite{SW20} showed that one
can obtain a $(1-\eps)$-approximation for directed \dsg via
$O(\log n/\eps)$ instances of \dsg with vertex weights. \cite{SW20}
claimed that their dynamic data structure for \dsg extended to
vertex-weighted graphs, and via their reduction, claimed a fully
dynamic $(1-\eps)$-approximate algorithm for directed \dsg. We
encountered some technical difficulties while trying to extend our
data structure to the vertex-weighted setting; we were also unable to
verify the correctness of the data structure in \cite{SW20} due to an
important missing technical detail. Bhattacharya et al. \cite{BHNT15}
describe a dynamic data structure for directed \dsg that maintains a
$(1/8-\eps)$-approximation in amortized polylogarithmic update time.
They also rely on a reduction to undirected graphs, but their
reduction is based on the one in \cite{ks-09} and loses a factor of
$2$ in the approximation unlike the one in \cite{SW20}.  In future
work we plan to address dynamic \dsg for vertex-weighted undirected
graphs and directed \dsg.

\paragraph{Organization:} In \refsection{prelim} we describe the connection between
fractional orientation and the exact LP relaxation for \dsg and discuss a
an approximate local optimal orientation that is crucial to our results.
In \refsection{amortized} we describe a simple amortized data structure. We build
upon it in \refsection{worst-case} to obtain a data structure with worst-case update time.
In \refsection{small-arboricity} we improve the worst-case update time by
data structures for the small arboricity and large arboricity regimes and exploiting the tradeoff.
We extend our data structure to hypergraphs in \refsection{extensions}.

\section{Preliminaries}
\labelsection{prelim} Let $G=(V,E)$ be an undirected multigraph.  Let
$A$ denote the set arcs obtained by bi-directing each edge. A
fractional orientation of $G$ is a function $y: A \rightarrow [0,1]$
such that for each undirected edge $e$ we have $y(a) + y(a') = 1$,
where $a$ and $a'$ are the bi-directed arcs corresponding to $e$.  $y$
is an integral orientation if $y(a) \in \{0,1\}$ for each arc.  We
call an integral orientation simply an orientation, and use $D$ to
denote the directed graph induced by the orientation. All quantities
discussed here will refer to the current state of the grant and
orientation as it evolves with edge insertions and deletions. For a
directed graph $D$ and a vertex $v$ we use $\delta^-(v)$ and
$\delta^+(v)$ to denote the set of incoming arcs into $v$, and the
outgoing arcs out of $v$ respectively.  We let $\inneighbors{v}$ to
denote the in-neighborhood of $v$ \wrt $D$; i.e., the endpoints $w$ of
arcs $(w,v) \in \incut{v}$ directed into $v$. Similarly, we let
$\outneighbors{v}$ to denote the out-neighborhood of $v$ \wrt $D$;
i.e., the endpoints $w$ of arcs $(v,w) \in \outcut{v}$ directed out of
$v$.

\subsection{Min-max orientations and an LP relaxation for \dsg}

We consider the LP formulation for \dsg from \cite{Charikar00}.
Recall that the objective is to find a set $S \subseteq V$ to maximize
the quantity $|E(S)|/|S|$. A natural way to write this is via
indicator variables $x_v, v \in V$ for inclusion in the optimum set
$S$.  An edge $e$ can be taken only if both $u$ and $v$ are in the set
$S$. Hence one can express the objective as
$\max \frac{\sum_{e = \edge{u}{v} \in E} \min\{x_u,x_v\}}{\sum_v x_v}$
with $x \in \{0,1\}^V$.\footnote{Here a pair $\edge{u}{v}$ will repeat
  in the sum according to the multiplicity of the edge.} To express
this via an LP relaxation, one can normalize the denominator with the
constraint $\sum_v x_v = 1$, and rewrite the convex objective
$\sum_{e = \edge{u}{v} \in E} \min{x_u,x_v}$ as a linear objective via
additional variables (which we omit). The LP and its dual are
described in \reffigure{dsglp}. The dual LP describes a fractional
orientation of $E$ to minimize the maximum in-degree of any
vertex. Here, for an edge $e$ and endpoint $v$, $y(e,v) \geq 0$
represents the fractional amount of $e$ directed towards $v$.  We let
$\lpopt$ denote the common optimum value of the above linear programs.
Charikar showed that the LP is an exact relaxation for \dsg and hence
$\lpopt = \rho(G)$.

\begin{figure}[t]
  \begin{framed}
    \begin{minipage}[t]{0.35\textwidth}
      \begin{align*}
        \text{maximize } & \sum_{e = \setof{u,v} \in E} \min{x_u, x_v} \\
        \text{over } & x_v \geq 0 \text{ for } v \in V \\
        \text{ s.t.\ } &
                         \sum_{v \in V} x_v \leq 1.
      \end{align*}
    \end{minipage}
    \hfill
    \vline
    \hfill
    \begin{minipage}[t]{0.5\textwidth}
      \begin{align*}
        \text{minimize }              %
        &                             %
          \max_v \sum_{e \in \delta(v)} y(e,v)
        \\
        \text{ over }                 %
        & y(e,v) \geq 0 \text{ for } e \in E \text{ and } v
          \in e
        \\
        \text{ s.t.\ }                %
        &                             %
          y(e,u) + y(e,v) \geq 1 \text{ for all } e = \setof{u,v}.
      \end{align*}
    \end{minipage}
  \end{framed}
  \caption{LP for \dsg and its dual.}
  \labelfigure{dsglp}
\end{figure}

\iffalse
  Let $A = \setof{(u,v), (v,u) \where \setof{u,v} \in E}$ denote the set
  arcs bi-directing each edge. A ``fractional direction'' of $E$ can be
  encoded as an assignment $y: A \to [0,1]$ such that
  $y(u,v) + y(v,u) = 1$ for all $\setof{u,v} \in E$.
\fi

Given a fractional orientation $y$, and a vertex $v$, we let
$\bary{v} \defeq \sum_{e \in \cut{v}} \y{e,v}$ denote the weighted
in-degree of $v$. The dual LP wants to minimize $\max_v \bary{v}$.

\subsection{Locally optimal and approximate locally optimal orientations}
\labelsection{local-analysis}\labelsection{local-optimality} A useful
and important idea for arboricity and density maintenance is the idea
of an (approximately) locally optimal orientation that has been
explored in previous work. We discuss some basics for the sake of
completeness before stating a specific approximate variant that we
work with.  Consider the following local constraint for $y$.
\begin{quote}
  \itshape For each edge $e = \setof{u,v}$, if $y(e,v) > 0$ then
  $\bary{v} \leq \bary{u}$.
\end{quote}
The idea is that if $y$ wants to minimize the maximum in-degree, then
it should never fractionally direct an edge towards an endpoint with
(strictly) larger in-degree. We say that $y$ is \defterm{locally
  optimal} when it satisfies the condition above for all arcs.
The following lemma can be shown.
\begin{lemma}
  There is an optimum solution to the dual LP satisfying the local
  optimality condition. Conversely, if $y$ satisfies the local
  optimality condition above, then $y$ is optimal.
\end{lemma}

We now consider approximations of the local optimality condition.  For
$\alpha, \beta \geq 0$, we say that $y$ is
\defterm{$(\alpha, \beta)$-locally optimal}, or a \defterm{local
  $(\alpha,\beta)$-approximation}, if for all edges $e = \setof{u,v}$,
\begin{quote}
  \itshape
  If $y(e,v) > 0$, then $\bary{v} < \alphamore \bary{u} + \beta$.
\end{quote}
Intuitively, the condition states that $y$ should never fractionally
an edge towards an endpoint with substantially larger in-degree.

This local optimality condition extends ideas from
\cite{KKPS14,SawlaniWang}. In terms of the definition above,
\cite{SawlaniWang} considered $(0,\eps \lambda)$ where $\lambda$ is a
constant factor estimate of the optimum density, which is inspired by
\cite{KKPS14} who consider $(0,c)$ for some fixed constant $c$. A
simple but key idea in our work is to introduce the multiplicative
dimension given by $\alpha > 0$. Following similar ideas in previous
work we show that approximate local optimality implies approximate
global optimality.

\begin{lemma}
  \labellemma{apx-local=>global}\labellemma{local-optimality} Let
  $\alpha, \beta > 0$ with $\alpha < c / \log n$ for a sufficiently
  small constant $c > 0$.  Let $\mu = \max_v \bary{v}$.  Let
  $k = \roundup{\log[1+\eps]{n}}$, and suppose that every arc $(u,v)$
  is $(\alpha, \beta)$-locally optimal.  Then
  \begin{align*}
    \mu \leq e^{\bigO{\sqrt{\alpha \log{n}}}} \parof{\lpopt +
    \bigO{\sqrt{\frac{\log{n}}{\alpha}}} \beta}
  \end{align*}
  In particular, given $\eps \in (0,1)$, for $\beta = \bigO{1}$ and
  $\alpha = c \eps^2 / \log{n}$ for a sufficiently small constant $c$,
  then the claim is as follows: if every arc $(u,v)$ is
  $(c \eps^2/\log n, \bigO{1})$-locally optimal, then
  $\mu \leq \epsmore \lpopt + \bigO{\ln{n} / \eps}$.
\end{lemma}

\begin{proof}
  We define an increasing sequence $\mu_0 < \mu_1 < \cdots$ where
  \begin{math}
    \mu_0 = 0
  \end{math}
  and
  \begin{math}
    \mu_i = \alphamore \mu_{i-1} + \beta
  \end{math}
  for $i \geq 1$.  Let $k \in \nnintegers$ be the unique index such
  that
  \begin{math}
    \mu_{k-1} \leq \mu < \mu_{k}.
  \end{math}
  Observe that
  \begin{align*}
    \mu_{k-i} \geq \frac{\mu}{\alphamore^i} - i \beta
    \geq e^{-\alpha i} \mu - i \beta
  \end{align*}
  for each $i \in \setof{1,\dots,k}$.

  For each index $i \in \setof{1,\dots,k}$, let
  $S_i = \setof{v \where \indegree{v} \geq \mu_{k-i}}$.  Note that
  $S_i$ is nonempty for all $i \leq k$.  Additionally, by
  $(\alpha,\beta)$-local optimality, we have
  \begin{math}
    \inneighbors{S_i} \subseteq S_{i+1}
  \end{math}
  for each $i$.

  Now, let $\eps > 0$ be a sufficiently small parameter to be chosen
  later. Since $S_1$ is non-empty and $\sizeof{S_i} \leq n$ for all
  $i$, there must be an index $i \leq \bigO{\log{n} / \eps}$ such that
  $ \sizeof{S_{i+1}} \leq \epsmore \sizeof{S_{i}}$.

  Consider the subgraph induced by $S_{i+1}$.  Every vertex in $S_{i}
  \subseteq S_{i+1}$
  is the head of at least
  \begin{align*}
    e^{- \alpha i} \mu - i \beta \geq e^{- \bigO{\alpha \log{n}
    / \eps}} \mu - \bigO{\frac{\log{n}}{\eps}} \beta
  \end{align*}
  fractional edges in the orientation, and the underlying undirected
  edges are all contained in the subgraph induced by
  $S_{i+1}$. Therefore we have
  \begin{align*}
    \lpopt &\geq
    \frac{\sizeof{E(S_{i+1})}}{\sizeof{S_{i+1}}}
    \geq
    \parof{e^{- \bigO{\alpha \log{n}
    / \eps}} \mu - \bigO{\frac{\log{n}}{\eps}} \beta} %
             \frac{\sizeof{S_{i}}}{\sizeof{S_{i+1}}}           %
             \\
           &\geq                       %
             \frac{1}{1+\eps}
             \parof{e^{-\bigO{\alpha \log{n} / \eps}} \mu -
             \bigO{\frac{\log{n}}{\eps}} \beta}.
  \end{align*}
  Rearranging,
  \begin{align*}
    \mu
    &\leq e^{\bigO{\alpha \log{n} / \eps}}
      \parof{\parof{1 + \eps}\lpopt
      + \bigO{\frac{\log{n}}{\eps}} \beta}
      \leq                        %
      e^{\bigO{\alpha \log{n} / \eps} + \eps} \parof{\lpopt +
      \bigO{\frac{\log{n}}{\eps}} \beta}.
  \end{align*}
  For $\eps= \bigO{\sqrt{\alpha \log{n}}}$, we have
  \begin{align*}
    \mu \leq e^{\bigO{\sqrt{\alpha \log{n}}}} \parof{\lpopt +
    \bigO{\sqrt{\frac{\log{n}}{\alpha}}} \beta},
  \end{align*}
  as desired.
\end{proof}

The proof also shows that it is easy to extract an approximate densest
subgraph from $y$; it is always a prefix of the list of vertices in
descending order of in-degree.  Moreover, identifying a prefix reduces
to keeping track of the cardinalities of the sets $S_i$ as defined in
the proof, and identifying an index $i$ such that
$\sizeof{S_{i+1}} \leq \epsmore \sizeof{S_i}$.

The data structures described in the rest of this article maintain
integral orientations of undirected and unweighted graphs. It will be
clear from their description that it is easy to maintain a list of the
vertices in decreasing order of in-degrees, as well as the
cardinalities of the sets $S_i$ so that we always know which prefix of
the list induces an approximate densest subgraph. We refer to this as
an \emph{implicit} representation of the approximate densest subgraph;
in particular, we can list off the vertices in an approximate densest
subgraph in $\bigO{1}$ time per vertex.  Later on, to simplify the
presentation of these data structures, we will focus on the aspects
maintaining the orientation is nearly optimal maximum in-degree rather
than on aspects of maintaining an approximate densest subgraph. We
assume that an implicit list representation of an approximate densest
subgraph, as described above, is maintained in the background with no
significant
overhead.

\section{Data structure with amortized update time guarantee}
\labelsection{amortized-unweighted}
\labelsection{amortized}

In this section we describe a simple fully dynamic data structure for
maintaining $(1-\eps)$-approximate densest subgraph that has
polylogarithmic \emph{amortized} update time.  It is based on
maintaining an approximate orientation of the graph. Here we recall
that $\lpopt$ refers the optimum fractional orientation and
$\arboricity{G}$ refers the optimum integral orientation. We have
$\lpopt \leq \arboricity{G} \leq \lpopt + 1$. Our data structures work
with integral orientations, and they have additive error when
comparing with $\lpopt$ or $\arboricity{G}$. We address the additive
error after the following theorem.

% The additive error makes
% the difference of $1$ between $\lpopt$ and $\arboricity{G}$
% negligible, so we do not emphasize the distinction here and
% throughout.

\begin{theorem}
  \labeltheorem{amortized-densest-subgraph} Consider the task of
  approximating the densest subgraph in an unweighted graph
  dynamically updated by edge insertions and deletions. Let $\eps > 0$
  be given.  Then one can maintain an orientation (explicitly) with
  maximum in-degree at most
  \begin{math}
    \epsmore \lpopt + \bigO{\log{n} / \eps},
  \end{math}
  and a subgraph (implicitly) with density at least
  \begin{math}
    \epsless \lpopt - \bigO{\log{n} / \eps},
  \end{math}
  in $\bigO{\log{n} / \eps^2}$ amortized time per update.  The data
  structure uses $\bigO{m+n}$ space.
\end{theorem}

\Cref{theorem:amortized-densest-subgraph} gives a
$\epsless$-approximation in $\bigO{\log{n} / \eps^2}$ amortized time
when the density is at least $\bigOmega{\log{n} / \eps^2}$. This
regime may already be of interest for many applications of densest
subgraph. To obtain an unconditional $\epsless$-approximation, one may
simply duplicate each edge $\bigO{\log{n} / \eps^2}$ times (as done by
\cite{SawlaniWang}), which ensures the density is sufficiently large,
while increasing the time for each edge insertion and deletion
multiplicatively by the same factor. We note that to maintain linear
space usage, one needs to make minor modifications so that the copies
of an edge use the same auxiliary data. We address these changes at
the end of this section (in \refsection{amortized-duplicate-space}).

\begin{corollary}
  \labelcorollary{amortized-densest-subgraph} Under the same
  conditions of \reftheorem{amortized-densest-subgraph}, one can
  maintain a subgraph with density at least $\epsless \lpopt$ in
  $\bigO{\log{n}^2 / \eps^4}$ amortized time per update and $\bigO{m+n}$
  space.
\end{corollary}

\subsection{High-level overview}

Let $\alpha = c \eps^2 / \ln{n}$, for a sufficiently small constant
$c$. We present a data structure that tries to maintain a
$\parof{1 + \bigO{\alpha}, \bigO{1}}$-locally optimal orientation as
edges are inserted and deleted. At a high-level, it automatically
flips an arc whenever it detects that local optimality for that arc is
no longer satisfied. In designing such a data structure there are two
high-level concerns. The first is to develop an organizing system to
efficiently detect arcs that violate the inequality. The second is to
control or account for the running time spent on ``cascades'', where
flipping one arc leads to violating local optimality for other
adjacent arcs, hence further arc flips.

\paragraph{Arc labels.}
The main ingredient in the data structure, and the only auxiliary data
stored in the data structure beyond the orientation itself, is a set
of integer \emph{endpoint labels} $\lab{u}{a}$ and $\lab{v}{a}$ for
each arc $a = (u,v)$. The labels play a key role both in maintaining
the local optimality conditions, and as a foothold for an amortized
analysis that can account for cascades of flips.

When an arc $a = (u,v)$ is added to the orientation, we record the
values of $\indegree{u}$ and $\indegree{v}$ (just after adding the arc
$a$) as $\lab{u}{a}$ and $\lab{v}{a}$, respectively. Periodically the
data structure resets $\lab{u}{a}$ and $\lab{v}{a}$ to the current
values of $\indegree{u}$ and $\indegree{v}$.

As mentioned above, the first role of the labels is to help maintain
and certify local optimality.  As we will show in
\refsection{amortized-optimality}, the data structure maintains
$\lab{u}{a}$ and $\lab{v}{a}$ such that
\begin{align*}
  \lab{u}{a} \leq \alphamore \indegree{u} + \bigO{1}
  \andcomma
  \indegree{v} \leq \alphamore \lab{v}{a} + \bigO{1}
  \text{, and }
  \lab{u}{a} \leq \lab{v}{a} + 1.
  \labelthisequation{amortized-label-invariants}
\end{align*}
Combining these inequalities implies
$\parof{1 + \bigO{\alpha}, \bigO{1}}$-local optimality.

The second role is to help amortize the time spent processing the
arcs, particularly in the presence of cascades. At a high-level, the
data structure only does work on an arc $a = (u,v)$ if $\indegree{u}$
is much smaller than $\lab{u}{a}$ or $\indegree{v}$ is much larger
than $\lab{v}{a}$. Meanwhile $\lab{u}{a}$ and $\lab{v}{a}$ reflect the
values of $\indegree{u}$ and $\indegree{v}$ at an earlier amount of
time. Thus we only process an arc $a$ after $\indegree{u}$ or
$\indegree{v}$ has deviated substantially from the point when the
labels for $a$ were set. These observation translates to an amortized
running time via a charging scheme described in
\refsection{amortized-running-time}.

\begin{figure}[t]
  \begin{framed}
    \small
    \begin{algorithm}{insert}{$e = \setof{u,v}$}
    \item We assume $\indegree{v} \leq \indegree{u}$. (Otherwise swap
      $u$ and $v$.)
    \item Add $a = (u,v)$ to the orientation, set
      $\lab{v}{a} = \indegree{v}$, set $\lab{u}{a} = \indegree{u}$ and
      call \algo{check-inc($v$)}.
    \end{algorithm}

    \smallskip

    \begin{algorithm}{delete}{$e = \setof{u,v}$}
    \item We assume $e$ is oriented as $a = (u,v)$.
    \item Delete $a$ from the orientation and call \algo{check-dec($v$)}.
    \end{algorithm}

    \smallskip

    \begin{algorithm}{check-inc}{$v$}
      \begin{blockcomment}
        We call this routine whenever $\indegree{v}$ has increased
        (always by $1$).
      \end{blockcomment}
    \item \labelstep{check-inc-loop} While there are arcs
      $a = (u,v) \in \incut{v}$ s.t.\
      \begin{math}
        \indegree{v} > \alphamore \lab{v}{a} + 1
      \end{math}
      \begin{steps}
        % \begin{blockcomment}
        %   $\incut{v}$ has gained at least
        %   $\bigOmega{\max{1,\alpha \indegree{v}}}$ arcs since $\lab{a}$
        %   was last updated.
        % \end{blockcomment}
      \item If $\indegree{u} < \indegree{v}$:
        \begin{steps}
        \item \labelstep{aic-inc} \labelstep{ainc-flip} Flip $a$ to
          $(v,u)$, and set $\lab{u}{a} = \indegree{u}$ and
          $\lab{v}{a} = \indegree{v}$.
          \begin{blockcomment}
            This also restores $\indegree{v}$ to its previous value
            and fixes the invariant for all arcs in $\incut{v}$.
          \end{blockcomment}
        \item Recurse by calling \algo{check-inc($u$)}, and return.
        \end{steps}
      \item Otherwise set $\lab{u}{a} = \indegree{u}$ and
        $\lab{v}{a} = \indegree{v}$.
      \end{steps}
      % \begin{blockcomment}
      %   If we reach this point outside the loop, then we were unable to
      %   decrease $\bary{v}$ by flipping arcs. For the amortized analysis:
      %   Spread $\bigO{1/\alpha}$ tokens uniformly over $\incut{v}$.
      % \end{blockcomment}
    \end{algorithm}

    \smallskip

    \begin{algorithm}{check-dec}{$u$}
      \begin{blockcomment}
        We call this routine whenever $\indegree{u}$ has decreased
        (always by $1$).
      \end{blockcomment}
    \item \labelstep{check-dec-loop} While there is an arc
      $a = (u,v) \in \outcut{u}$ s.t.\
      $\lab{u}{a} > \alphamore \indegree{u} + 1$
      \begin{steps}
      \item If $\indegree{u} < \indegree{v}$:
        \begin{steps}
        \item Flip $a$ to $(v,u)$. Set $\lab{u}{a} = \indegree{u}$ and
          $\lab{v}{a} = \indegree{v}$.
          \begin{blockcomment}
            This also restores $\indegree{u}$ to its previous value.
          \end{blockcomment}
        \item Recurse by calling \algo{check-dec($v$)}, and return.
        \end{steps}
      \item Otherwise set $\lab{u}{a} = \indegree{u}$ and
        $\lab{v}{a} = \indegree{v}$.
      \end{steps}
      % \begin{blockcomment}
      %   Else, for the amortized analysis: Spread $\bigO{1/\alpha}$
      %   tokens uniformly over $\incut{v}$.
      % \end{blockcomment}
    \end{algorithm}
  \end{framed}
  \vspace{-1em}
  \caption{Dynamically approximating the min-max orientation in an
    unweighted graph with fast amortized update times.
    \labelfigure{unweighted-amortized-code}}
\end{figure}

\paragraph{The data structure:} Pseudocode for the data structure is
presented in \reffigure{unweighted-amortized-code}. Clearly it is very
simple. At a high-level, the data structure adds and deletes arcs as
requested and then makes local flips and resets arc labels to repair
the inequalities in \refequation{amortized-label-invariants} whenever
they are violated. When inserting an edge $e = \setof{u,v}$, it
orients $e$ towards the vertex with smaller in-degree. When deleting
an edge $e$, it removes the corresponding oriented arc. These
operations increase or decrease the in-degree $\indegree{v}$ of an
endpoint $v$, and may violate the inequalities in
\refequation{amortized-label-invariants} above, which relate
$\indegree{v}$ to the labels $\lab{v}{a}$ for arcs $a$ in $\incut{v}$
or $\outcut{v}$. To check and repair these inequalities we introduce
two subroutines \algo{check-inc($v$)} and \algo{check-dec($v$)}.

We call \algo{check-inc($v$)} whenever the in-degree of a vertex $v$
is increased. The subroutine checks for any arcs
$a = (u,v) \in \incut{v}$ where $\indegree{v}$ has become too large
relative to $\lab{v}{a}$. For each such arc $a$, depending on whether
or not $\indegree{u} < \indegree{v}$, it either flips $a$ (restoring
$\indegree{v}$ to its previous value) and resets $\lab{v}{a}$, or
relabels $\lab{v}{a} = \indegree{v}$. A flip would increase
$\indegree{u}$, so in this case we recurse on $u$.

The other routine, \algo{check-dec($u$)}, is similar to
\algo{check-inc} except it is for the case where the in-degree of a
vertex $u$ is decreased. \algo{check-dec($u$)} makes sure that
$\indegree{u}$ is not too much smaller than $\lab{u}{a}$ for any arc
$a \in \outcut{u}$. When violations are found, we either flip the
violating arc or reset its label. A flip leads to a recursive call to
\algo{check-dec} on the opposite endpoint, hence possibly more flips.

The point of the calls to \algo{check-inc} and \algo{check-dec} is to
ensure that the label inequalities are met for all arcs in the
orientation. We call these subroutines appropriately whenever an
in-degree changes and a violation might be created. The conditional
loops in these subroutines ensure the subroutines do not terminate
until all violating labels are addressed.

\subsection{Maintaining a
  $\parof{1 + \bigO{\alpha}, \bigO{1}}$-locally optimal orientation}

\labelsection{amortized-optimality}

We now prove formally that the data structure maintains a
$\parof{1 + \bigO{\alpha}, \bigO{1}}$-locally optimal orientation. The
local optimality is certified via the arc labels as described above.

\begin{lemma}
  \labellemma{amortized-labels}
  For all arcs $a = (u,v)$, we have
  \begin{math}
    \lab{u}{a} \geq \lab{v}{a} - 1.
  \end{math}
\end{lemma}
\begin{proof}
  The labels $\lab{u}{a}$ and $\lab{v}{a}$ are set only in two
  situations.  The first setting is when we add $a$ to the orientation,
  either upon inserting $\setof{u,v}$, or from flipping $(v,u)$ to
  $(u,v)$. When inserting $\setof{u,v}$, the choice of orientation
  implies that $\indegree{v} \leq \indegree{u}$ before adding $a$, so
  we have
  \begin{math}
    \lab{v}{a} \leq \lab{u}{a} + 1.
  \end{math}
  as desired.  When flipping $(v,u)$ to $(u,v)$, we have
  $\indegree{v} \leq \indegree{u} - 1$ before flipping, hence
  \begin{math}
    \lab{v}{a} \leq \lab{u}{a} + 1.
  \end{math}

  The second setting where we reset $\lab{u}{a}$ and $\lab{v}{a}$ is
  after we choose \emph{not} to flip $(u,v)$ in either
  \algo{check-inc($v$)} or \algo{check-inc($u$)}. In either case we
  would have just verified that $\indegree{u} \geq \indegree{v}$,
  hence
  \begin{math}
    \lab{u}{a} \geq \lab{v}{a}
  \end{math}
  as well.
\end{proof}

\begin{lemma}
  \labellemma{amortized-label-v}
  For all arcs $a = (u,v)$, we have
  \begin{math}
    \indegree{v} \leq \alphamore \lab{v}{a} + 1.
  \end{math}
\end{lemma}

\begin{proof}
  Fix $a$. If $\indegree{v}$ momentarily increases as to violate the
  desired inequality, then in the subsequent call to
  \algo{check-inc($v$)}, the data structure will continue to process
  edges in $\incut{v}$ until either (a) it flips some arc (possibly
  $a$) in $\incut{v}$ or (b) resets $\lab{v}{a}$. In event (a),
  $\indegree{v}$ is decreased to its previous value before the
  inequality was violated. In event (b), we set
  $\lab{v}{a} = \indegree{v}$ which satisfies the inequality.
\end{proof}

\begin{lemma}
  \labellemma{amortized-label-u}
  For all arcs $a = (u,v)$,
  \begin{math}
    \lab{u}{a} \leq \alphamore \parof{\indegree{u} + 1}.
  \end{math}
\end{lemma}

\begin{proof}
  The proof is similar to the proof of \reflemma{amortized-label-v}.
  Fix $a$. If $\indegree{u}$ momentarily decreases as to violate the
  desired inequality, then in the subsequent call to
  \algo{check-dec($u$)}, the data structure will continue to pull
  process edges in $\outcut{u}$ until either (a) it flips some arc
  (possibly $a$) in $\outcut{u}$ or (b) resets $\lab{u}{a}$. In event
  (a), $\indegree{u}$ is increased to its previous value before the
  inequality was violated. In event (b), we set
  $\lab{u}{a} = \indegree{u}$ which satisfies the inequality.
\end{proof}

% \begin{proof}
%   When $a$ is initially added to the orientation we have
%   $\lab{u}{a} = \indegree{u}$ as desired. If $\indegree{u}$ decreases
%   and momentarily violates the inequality, the decrease is immediately
%   followed by a call to \algo{check-dec($u$)}. \algo{check-dec($u$)}
%   processes arcs $b \in \outcut{u}$ whose labels also violate the
%   inequality in question. For each such arc (which may be $a$), it
%   either flips $b$ or resets $\lab{u}{b}$ to $\indegree{u}$. In the
%   first event, the $\indegree{u}$ is increased to its previous
%   value. In the second event, if $b = a$, then the new value of
%   $\lab{u}{a}$ fixes the inequality. If $b$ is not $a$, then the loop
%   continues until it eventually corrects the situation for $a$.
% \end{proof}

\begin{lemma}
  \labellemma{amortized-local-optimality}
  The orientation is always
  $\parof{1 + \bigO{\alpha}, \bigO{1}}$-locally optimal.
\end{lemma}
\begin{proof}
  Fix $a = (u,v)$. By
  \cref{lemma:amortized-labels,lemma:amortized-label-u,lemma:amortized-label-v}
  we have
  \begin{align*}
    \indegree{v} \leq \alphamore \lab{v}{a} + 1 \leq \alphamore
    \lab{u}{a} + 2 + \alpha \leq \alphamore^2 \indegree{u}{a} + 3 + 3 \alpha,
  \end{align*}
  as desired.
\end{proof}

\subsection{Running time analysis}

\labelsection{amortized-running-time}

The description above establishes that the data structure maintains a
$\parof{1 + \bigO{\alpha}, \bigO{1}}$-locally optimal orientation,
which implies global optimality.  In this section we address the
remaining issue of (amortized) running time. We mention that in
addition to the work described in the pseudocode, the data structure
also maintains a list representation of the vertices of an approximate
densest subgraph in the background. As discussed at the end of
\refsection{local-optimality}, this is fairly simple to do with
negligible overhead as it largely consists of maintaining a list of
the vertices in decreasing order of in-degree. We have omitted these
details from the pseudocode as we feel they distract from the main
points of the analysis.

We now focus on analyzing the algorithm pertaining to the
pseudocode. To simplify the discussion we first explain how, with some
simple auxiliary data structures, each step in the pseudocode takes
constant time.

In particular, we explain how to organize the arcs so that in the
loops of \algo{check-inc} and \algo{check-dec}, each arc $a$ can be
generated in $\bigO{1}$ time. For each vertex $v$, we maintain the
arcs in $\incut{v}$, and the arcs in $\outcut{v}$, in order of
$\lab{v}{a}$, in two nested doubly linked lists. We first describe the
construction for $\incut{v}$. We place each arc in $\incut{v}$ in a
doubly linked list consisting of all arcs $a$ in $\incut{v}$ with the
same label $\lab{v}{a}$. We then place these lists in an outer doubly
linked list, in order of label. We also maintain maintain a pointer to
the location of the first list of arcs $a$ with label $\lab{v}{a}$
greater than equal to $\indegree{v}$.  This allows for the following
constant time operations. First, we can retrieve the minimum or
maximum label in constant time. Second, we can insert a new arc
$a \in \incut{v}$ with label $\label{v}{a} = \indegree{v}$ in constant
time.

For $\outcut{u}$, we construct the same data structure as described
above except with respect to the labels $\lab{u}{a}$.

With the arcs in $\outcut{v}$ and $\incut{v}$ sorted as described
above, we can make each iteration of the loops in \algo{check-inc} and
\algo{check-dec} run in $\bigO{1}$ time by querying these data
structures for the minimum or maximum label arc. When updating
$\lab{u}{a}$ or $\lab{v}{a}$ for an arc $a = (u,v)$, since these
labels are set to $\indegree{u}$ and $\indegree{v}$, we can use our
additional pointers to insert them into the appropriate lists in
$\bigO{1}$ time.

We now move onto the amortized analysis of the data structure with the
understanding that each line of the pseudocode takes constant time.
The main issue is that the recursive calls in \algo{check-inc} and
\algo{check-dec} can potentially lead to many flips for a single
insertion and edge deletion, and the challenge is to amortize these
flips. At a high-level, the analysis observes that an arc $a = (u,v)$
is processed only when $\indegree{u}$ or $\indegree{v}$ have deviated
substantially from the labels $\lab{u}{a}$ and $\lab{v}{a}$,
respectively.  We devise a charging scheme that allows us to amortize
the time processing $a$ against the change to $\indegree{u}$ or
$\indegree{v}$. Meanwhile each edge insertion or edge deletion (after
accounting for the full chain of local flips) ultimately changes the
in-degree of a single vertex by $1$, which is reflected in the
amortized cost.

\begin{lemma}
  \labellemma{amortized-running-time}
  \labellemma{unweighted-amortized-update} \algo{insert} and
  \algo{delete} take $\bigO{1/\alpha} = \bigO{\log{n} / \eps^2}$
  amortized time.
\end{lemma}
\begin{proof}
  Observe that the net effect of \algo{insert} on the vertex
  in-degrees is to increase $\indegree{x}$ of a single vertex $x$ by
  $1$. Similarly, \algo{delete} decreases $\indegree{x}$ of a single
  vertex $x$ by $1$.

  Now, the running time in \algo{insert} is proportional to the number
  of arcs considered in the while loop (step \refstep{check-inc-loop})
  of \algo{check-inc}, over all recursive calls to \algo{check-inc}
  (plus a constant amount of work). The running time in \algo{delete}
  is proportional to the number of arcs considered in the while loops
  of \algo{check-dec}.

  Our amortized analysis is a fractional charging scheme, where we
  distribute fractional credits to each arc that accumulate and pay
  for processing the arc later. The credits are generated as follows.
  Whenever either \algo{insert} or \algo{delete} results in changing
  the in-degree of a vertex $x$ (by $1$), we spread
  $\bigO{1 / \alpha}$ credits uniformly over the arcs in $\incut{x}$;
  thus each arc $a \in \incut{x}$ receives a credit of
  $ \bigOmega{1 / \alpha \indegree{x}}$. One unit of credit will pay
  for a constant amount of work, so this adds an amortized cost of
  $\bigO{1/\alpha}$ to both \algo{insert} and \algo{delete}.

  \paragraph{Claim:}\emph{Any arc $a$ processed by \algo{check-inc} or
    \algo{check-dec} has acquired at least $1$ unit of credit since
    $\lab{a}$ was last set.}

  \medskip

  Consider first \algo{check-inc}; suppose $a = (u,v)$ is processed in
  the loop. Since
  \begin{math}
    \indegree{v} > \alphamore \lab{v}{a} + 1,
  \end{math}
  and $\lab{v}{a}$ had been set to $\indegree{v}$ earlier,
  $\indegree{v}$ must have gained at least
  $\bigOmega{\alpha \indegree{v}}$ edges since $a$ was last
  labeled. Each edge gained by $\incut{v}$ contributes
  $\bigOmega{\frac{1}{\alpha \indegree{v}}}$ credits to $a$ and thus
  $a$ has one unit of credit to pay by the time $\indegree{v}$
  increases to greater than $\alphamore \lab{v}{a} + 1$. This proves
  the part of the claim concerning \algo{check-inc}.

  Now consider \algo{check-dec($u$)}; suppose $a = (u,v)$ is processed
  in the loop. Recall that $\lab{u}{a}$ had been set to $\indegree{u}$
  earlier; now (when $a$ is processed) we have
  $\lab{u}{a} > \alphamore \indegree{u} + 1$. Therefore $\incut{u}$
  must have lost at least $\bigOmega{\max{1, \alpha \indegree{u}}}$
  arcs since $\lab{u}{a}$ was set. Each edge lost contributes
  $\bigOmega{1/ \alpha \max{1,\indegree{u}}}$ credits to $a$.
  Multiplying these quantities together shows that $a$ has acquired at
  least one credit before it is processed, as claimed.

  Now, the claim implies that the constant work in every iteration of
  the loop in \algo{check-inc} (except the last, over all recursive
  calls), as well as for every recursive call in \algo{check-dec}, is
  paid for by existing credit.  All put together, the overall
  amortized running time of each operation is bounded above by the
  initial amortized cost, $\bigO{1/\alpha}$.
\end{proof}

\begin{remark}
  Of course one could have the insertion of an edge $e$ pay for the
  amortized cost of deleting $e$ later, and claim that deletion takes
  $\bigO{1}$ amortized time. We do not emphasize this distinction.
\end{remark}

\subsection{Extending to fractional orientations.}\labelsection{amortized-duplicate-space}
As discussed above, to obtain a proper $\epsmore$-approximation to the
fractional arboricity, one can duplicate each edge
$k = \bigO{\log{n} / \eps^2}$ times for sufficiently large $C > 0$,
and interpret each ``duplicate'' as a fractional edge of weight
$1/k$. This increases the running time of all operations by
$\bigO{\log{n} / \eps^2}$. Additionally it would increase the space by
a $\bigO{\log{n}/\eps^2}$-factor. We would like to avoid this
additional space overhead and here we will explain how to
\emph{simulate} the duplication approach in linear space.

Let $e = \setof{u,v}$ be a fixed edge. In an orientation, $e$ is
directed as either $a_1 = (u,v)$ or $a_2 = (v,u)$. If we duplicate $e$
$k$ times, then some duplicates will be of the form $a_1$ and the rest
will be of the form $a_2$.

Now, rather than record each copy separately, we can instead record
numerically how many copies of $e$ are oriented in each direction.
Additionally, for each orientation $a_i$ of $e$, we will maintain one
pair of labels $\lab{u}{a_i}$ and $\lab{v}{a_i}$ that serve all copies
of that arc, rather than having each copy of $a_i$ have its own set of
labels. Thus when the data structure resets the labels for one copy of
the arc $a_i$, this automatically resets the labels for all copies of
the arc $a_i$ simultaneously.

We argue that resetting the labels of all the copies of an arc $a_i$,
rather than a particular copy, still preserves correctness. In
general, we reset the labels $\lab{u}{a_i}$ and $\lab{v}{a_i}$ when
doing so would preserve the inequalities in
\cref{lemma:amortized-labels,lemma:amortized-label-u,lemma:amortized-label-v}. In
particular, if it is valid to update $\lab{u}{a_i}$ and $\lab{v}{a_i}$
for one particular copy of $a_i$, then it is valid to update the
labels for all copies of $a_i$. Thus no error is introduced; if
anything, updating the labels of the copies of $a_i$ can be understood
as ``free'' updates that only help the data structure.

The final point to address is for the nested lists that maintain the
arcs in $\incut{v}$ and $\outcut{v}$ in order of $\label{v}{a}$. Here
we take advantage of the fact that all copies of an arc have the same
label. Consider an arc $a \in \incut{v}$ (say). Rather than store each
copy of $a$ separately, we have a node representing $a$ along with the
number of copies of $a$ that are in $\incut{v}$. Since all the copies
of the same arc have the same label, they would occupy the same place
in the list anyway. Now removing (a copy of) an arc from the list
corresponds to decrementing the corresponding counter, unless it was
the last copy in which case the node for that arc is removed. Likewise
inserting an arc corresponds to incrementing a counter unless it is
the first copy in which case a new node is created.

\section{$(1+\eps, \log{n})$-approximate orientation with worst-case
  updates}

\labelsection{worst-case-unweighted}
\labelsection{worst-case}

The previous section gives a data structure that dynamically maintains
a $\parof{1 + \eps, \log{n}/\eps}$-approximate orientation in
$\bigO{\log{n} / \eps^2}$ amortized time per update. This section
extends that data structure to obtain
$\bigO{\log{n}^2 \log{\lpopt + \log{n} / \eps} / \eps^4}$
worst-case time per update while retaining the same amortized running
time. Formally, we will prove the following.

\begin{theorem}
  \labeltheorem{worst-case-mixed-apx}\labeltheorem{wc-mixed-apx} Let
  $G$ be an unweighted and undirected graph over $n$ vertices,
  dynamically updated by edge insertions and deletions. Then one can
  maintain an orientation of $G$ with maximum in-degree at most
  $\epsmore \lpopt + \bigO{\log{n} / \eps}$, and (implicitly) a
  subgraph with density at least
  $\epsless \lpopt - \bigO{\log{n} / \eps}$, in
  $\bigO{\log{n} / \eps^2}$ amortized time and
  $\bigO{\log{n}^2 \log{\lpopt + \log{n} / \eps} / \eps^4}$ worst-case
  time per edge insertion or deletion. The data structure uses
  $\bigO{m + n}$ space.
\end{theorem}

As before, we can also obtain a fractional orientation of maximum load
at most a $\epsmore$-factor of the fractional arboricity by implicitly
duplicating each edge $\bigO{\log{n} / \eps^2}$ times. Again, to
maintain linear space storage, some minor adjustments are required so
that the ``copies'' of an edge use the same data. We briefly comment
on the adjustments at the end of this section
(\refsection{worst-case-duplicate-space}). Altogether we obtain the
following bounds that increase the running times in
\reftheorem{worst-case-mixed-apx} by roughly a
$\bigO{\log{n} / \eps^2}$-factor.
\begin{corollary}
  \labelcorollary{worst-case-apx} Let $G$ be an unweighted and
  undirected graph over $n$ vertices, dynamically updated by edge
  insertions and deletions. Then one can maintain a fractional
  orientation of $G$ with maximum in-degree at most $\epsmore \lpopt$,
  and (implicitly) a subgraph with density at least $\epsless \lpopt$,
  in $\bigO{\log{n}^2 / \eps^4}$ amortized time and
  $\bigO{\log{n}^3 \parof{\log{\lpopt} + \log \log{n} + \log{1/\eps}}
    / \eps^6}$ worst-case time per edge update. The data structure
  uses $\bigO{m+n}$ space.
\end{corollary}

\subsection{High-level overview}

As mentioned above, the new data structure is based on the amortized
data structure form the previous section.  To motivate the changes we
first explain where the previous (purely) amortized approach can have
bad worst-case performance.  There are two factors that are unbounded:
the depth of recursive calls to \algo{check-inc} or \algo{check-dec},
and the length of the loop within a single call to \algo{check-inc} or
\algo{check-dec}. We discuss them separately and consider the
recursive aspect first.  Consider the subroutine \algo{check-inc} in
\reffigure{unweighted-amortized-code}. Each time we flip an arc in
step \refstep{ainc-flip} we also make a recursive call to the opposite
endpoint, which may trigger further flips and recursive calls. The
total number of flips can potentially be very large. Similarly
\algo{check-dec} can have many flips via recursive calls.

We curtail this scenario by increasing the requirements for a
flip. Before, we always flipped an arc from $(u,v)$ to $(v,u)$ if $u$
is the smaller degree vertex. Now we only flip if $u$ is substantially
smaller than $v$: namely, only if
$\indegree{v} \geq \alphamore \parof{\indegree{u} + 1}$.  For
\algo{check-inc($v$)} this has the following effect. When recursing to
$u$, we know that $\indegree{u}$ is smaller than $\indegree{v}$ was
(at the time of calling \algo{check-inc($v$)}) by a
$\alphamore$-factor. Meanwhile the in-degrees are integral, bounded
above by $\bigO{\arboricity{G} + \log{n} / \eps}$ (pending proof of
correctness) and bounded below by $0$. Therefore there are at most
$\bigO{\log[1 + \alpha]{\arboricity{G} + \log{n} / \eps}} =
\bigO{\log{\arboricity{G} + \log{n} / \eps}/\alpha}$ recursive
calls. The recursion depth for \algo{check-dec} is bounded similarly;
here the degrees increase a $\alphamore$-multiplicative factor with
each recursive call.

We note that a similar argument as described above, except in purely
additive terms, limits the recursive depth in previous work
\cite{KKPS14,SW20}.

The second issue is that the loops may be very long when there are
many labels that require updating, but are not actually
flipped. (Flipping restores the in-degree to the previous value and
terminates the loop.) To try to limit the number of such label
updates, we adjust the data structure to process extra arcs to ``get
ahead'' of the expiring labels.  More precisely, we adjust the loop to
try to process $\bigO{1 / \alpha}$ arcs even if not all of them are
critically outdated.  This helps the data structure stay ahead of a
glut of labels about to expire. That said, we only process arc labels
that are at least a little outdated, and this allows us to retain the
amortized running time.  The formal description of the algorithms are
in \reffigure{unweighted-worst-case-code}.

\begin{figure}[t!]
  \begin{framed}
    \small
    \begin{algorithm}{check-inc}{$v$}
      \begin{blockcomment}
        We call this routine whenever $\indegree{v}$ has increased
        (always by $1$).
      \end{blockcomment}
    \item \labelstep{wci-loop} For up to $C/\alpha$ arcs
      $a = (u,v) \in \incut{v}$ s.t.\
      \begin{math}
        \indegree{v} \geq \parof{1 + \alpha / 2} \lab{v}{a},
      \end{math}
      in increasing order of $\lab{v}{a}$, for a sufficiently large
      constant $C$: % \commentcode{Revised}
      \begin{steps}
      \item If
        $\indegree{v} \geq \alphamore \parof{\indegree{u} + 1}$:
        \commentcode{Revised}
        \begin{steps}
        \item \labelstep{wc-check-inc-flip} Flip $a$ to $(v,u)$ and
          set $\lab{u}{a} = \indegree{u}$ and
          $\lab{v}{a} = \indegree{v}$.
          \begin{blockcomment}
            This restores $\indegree{v}$ to its previous value.
          \end{blockcomment}
        \item Call \algo{check-inc($u$)} and return.
        \end{steps}
      \item \labelstep{wc-increase-label} Otherwise set
        $\lab{v}{a} = \indegree{v}$ and $\lab{u}{a} = \indegree{u}$.
      \end{steps}
      % \begin{blockcomment}
      %   If we reach this point outside the loop, then we were unable to
      %   decrease $\bary{v}$ by flipping arcs. For the amortized analysis:
      %   Spread $\bigO{1/\alpha}$ tokens uniformly over $\incut{v}$.
      % \end{blockcomment}
    \end{algorithm}

    \smallskip

    \begin{algorithm}{check-dec}{$u$}
      \begin{blockcomment}
        We call this routine whenever $\indegree{u}$ has decreased
        (always by $1$).
      \end{blockcomment}
    \item For up to $C/\alpha$ arcs $a = (u,v) \in \outcut{u}$ s.t.\
      $\lab{u}{a} \geq \parof{1 + \alpha/ 2} \indegree{u}$, in
      decreasing order of $\lab{u}{a}$, for a sufficiently large
      constant $C$:
      \begin{steps}
      \item If $\indegree{v} \geq \alphamore \parof{\indegree{u} + 1}$:
        \begin{steps}
        \item Flip $a$ to $(v,u)$. Set $\lab{u}{a} = \indegree{u}$ and
          $\lab{v}{a} = \indegree{v}$.
        \item Call \algo{check-dec($u$)} and return.
        \end{steps}
      \item Otherwise set $\lab{v}{a} = \indegree{v}$ and
        $\lab{u}{a} = \indegree{u}$.
      \end{steps}
    \end{algorithm}
  \end{framed}
  \vspace{-1em}
  \caption{Revised implementations of \algo{check-inc} and
    \algo{check-dec} (cf.\ \reffigure{unweighted-amortized-code}) for
    fast worst-case update
    times.\labelfigure{unweighted-worst-case-code}}
\end{figure}

\subsection{Maintaining a
  $\parof{1 + \bigO{\alpha}, \bigO{1}}$-locally optimal orientation}

\labelsection{wc-optimality}

In this section we prove that the data structure maintains a
$\parof{1 + \bigO{\alpha}, \bigO{1}}$-locally optimal orientation. By
\reflemma{apx-local=>global} this implies the global optimality
conditions described in \reftheorem{worst-case-mixed-apx}.  The
overall analysis is structured similarly to that of the previous data
structure in \refsection{amortized-optimality}; in particular,
\cref{lemma:wc-labels,lemma:wc-label-u,lemma:wc-label-v,lemma:wc-local-optimality}
are in one-to-one correspondence with
\cref{lemma:amortized-labels,lemma:amortized-label-u,lemma:amortized-label-v,lemma:amortized-local-optimality}.

\begin{lemma}
  \labellemma{wc-labels} For every arc $a = (u,v)$, we have
  \begin{math}
    \lab{v}{a} \leq \alphamore \parof{\lab{u}{a} + 1}.
  \end{math}
\end{lemma}

\begin{proof}
  For any arc $a = (u,v)$, the algorithm sets $\lab{u}{a}$ and
  $\lab{v}{a}$ to $\indegree{u}$ and $\indegree{v}$ only after
  verifying that
  $\indegree{v} \leq \alphamore \parof{\indegree{u} + 1}$.
\end{proof}

    %     \begin{proof}
    %     Fix $a$.  We update $\lab{u}{a}$ and $\lab{v}{a}$ in only four
    %     situations.

    %     The first is after adding $a$ to the orientation when inserting the
    %     underlying undirected edge $\setof{u,v}$. Then we had
    %     $\indegree{v} \leq \indegree{u}$ before adding $a$, we set
    %     $\lab{u}{a} = \indegree{u}$ and $\lab{v}{a} = \indegree{v}$ after
    %     adding $a$. So we will have $\lab{v}{a} \leq \lab{u}{a} + 1$.

    %     The second case is when flipping from $(v,u)$ to $(u,v)$, in
    %     \refstep{wc-check-inc-flip} of a call to \algo{check-inc($u$)}. Then
    %     we had $\indegree{v} < \indegree{u}$ before flipping to $(u,v)$, and
    %     we set $\lab{u}{a} = \indegree{u}$ and $\lab{v}{a} = \indegree{v}$
    %     after flipping. So we have $\lab{u}{a} \leq \lab{v}{a} + 1$.

    %     The third, in \refstep{wc-decrease-label-u} of
    %     \algo{check-inc($v$)}, comes after verifying that
    %     $\indegree{v} < \alphamore \indegree{u} + 1$. Thus setting
    %     $\lab{u}{a} = \indegree{a}$ and $\lab{v}{a} = \indegree{v}$
    %     satisfies the desired inequality.

    %     The last situation comes at the end of \algo{check-dec($u$)}. But
    %     here we can only decrease $\lab{u}{a}$ which does not hurt the
    %     inequality.
    %     \end{proof}

\begin{lemma}
  \labellemma{wc-label-v}
  For every arc $a = (u,v)$, we have
  \begin{math}
    \indegree{v} \leq \alphamore \parof{\lab{v}{a} + 1}.
  \end{math}
\end{lemma}

\begin{proof}
  Fix $a$. Let us call $a$ \emph{bad} if it violates the claimed
  inequality, and \emph{dangerous} if
  \begin{align*}
    \indegree{v} \geq \parof{1 + \alpha / 2} \lab{v}{a}.
  \end{align*}
  We want to show that $a$ is never bad.  Now, for $a$ to be bad it
  must first be dangerous, which makes it eligible to be processed in
  the loop for \algo{check-inc($v$)}. After $a$ becomes dangerous,
  $\indegree{v}$ must still increase by
  $\bigOmega{\alpha \indegree{v}}$ before $a$ becomes bad, and each
  increment presents an opportunity to process $a$. We want to argue
  that $a$ will be processed before it comes bad.

  Recall that \algo{check-inc($v$)} processes the arc
  $b \in \incut{v}$ with minimum label $\lab{v}{b}$. Let us say an arc
  $\lab{v}{b}$ has \emph{higher priority} than $a$ if
  $\lab{v}{b} \leq \lab{v}{a}$.  When $a$ becomes dangerous, there are
  at most $\indegree{v}$ other arcs with higher
  priority. Additionally, as long as $a$ remains dangerous, there will
  be no new higher priority arcs because each new label is set to
  $\indegree{v}$, and $\indegree{v} > \lab{v}{a}$.

  Each time $\indegree{v}$ increases, if we do not process $a$, then
  we instead process $C/\alpha$ higher-priority arcs for a suitably
  large constant $C$. Each higher-priority arc $b$ is either flipped
  or has $\lab{v}{b}$ reset to be greater than $\lab{v}{a}$. Either
  way, the arc $b$ will no longer be higher-priority, and the number
  of higher-priority arcs has decreased by $1$.

  To recap, each unit increase to $\indegree{v}$ processes
  $C/\alpha$ higher-priority arcs, so we would process all
  higher-priority arcs before $\indegree{v}$ changes by enough to make
  $a$ bad. This forces us to process $a$ before $a$ becomes
  bad. Processing $a$ will either reset $\lab{v}{a}$ or flip $a$; in
  either case, $a$ will no longer be dangerous.
\end{proof}

\begin{lemma}
  \labellemma{wc-label-u} For every arc $a = (u,v)$, we have
  \begin{math}
    \lab{u}{a} \leq \alphamore \parof{\indegree{u} + 1}.
  \end{math}
\end{lemma}

\begin{proof}
  The proof is very similar to the proof of \reflemma{wc-label-v}.
  Fix $a$. Let us
  call $a$ \emph{bad} if it violates the claimed inequality, and
  \emph{dangerous} if
  \begin{align*}
    \indegree{u} \geq \parof{1 + \alpha / 2} \lab{u}{a}.
  \end{align*}
  Whenever $\lab{u}{a}$ is set, it is always set to $\indegree{u}$, in
  which case it is not dangerous. For $a$ to become bad, it must first
  be dangerous. If $a$ is dangerous, then it is eligible to be
  processed for the loop in \algo{check-dec($u$)} (which resets
  $\lab{u}{a}$). Additionally, after $a$ becomes dangerous,
  $\indegree{u}$ must still decrease by
  $\bigOmega{\alpha \indegree{u}}$ before $a$ becomes bad, and each
  increment is an opportunity to process $a$. We want to argue that
  $a$ must be processed before it comes bad.

  \algo{check-dec($u$)} repeatedly processes the arc $b \in
  \incut{u}$ with maximum label $\lab{u}{b}$. Let us say an arc
  $\lab{u}{b}$ has \emph{higher priority} than $a$ if
  $\lab{u}{b} \geq \lab{u}{a}$.  There are at most $\indegree{u}$
  other arcs with higher priority when $a$ becomes dangerous. As
  long as $a$ remains dangerous, there are no new higher
  priority arcs because each new label is set to $\indegree{u}$,
  and $\indegree{u} < \lab{u}{a}$.

  Each decrease in $\indegree{u}$ that does not process $a$ must
  process $C/\alpha$ higher-priority arcs for a suitably large
  constant $C$. Each processed arc is no longer higher-priority
  after processing, so the number of higher-priority arcs
  decreases by $1$ with each iteration.

  To recap, each unit decrease to $\indegree{u}$ processes and
  removes $C/\alpha$ higher-priority arcs.  There are only
  $\indegree{u}$ higher-priority arcs total and by the time
  $\indegree{u}$ decreases by $\bigOmega{\alpha \indegree{u}}$, all
  higher-priority arcs will be processed. Thus the data structure
  processes $a$ before $a$ becomes bad. Processing $a$ resets
  $\lab{u}{a}$ or flips $a$, and $a$ is no longer dangerous.
\end{proof}

\begin{lemma}
  \labellemma{wc-local-optimality}
  The data structure maintains
  $\parof{1+\bigO{\alpha}, \bigO{1}}$-local optimality.
\end{lemma}

\begin{proof}
  For each arc $a = (u,v)$, by
  \cref{lemma:wc-labels,lemma:wc-label-v,lemma:wc-label-u},
  we have
  \begin{align*}
    \indegree{v} \leq \alphamore \lab{v}{a} + 1 + \alpha
    \leq                        %
    \alphamore^2 \lab{u}{a} + 2 + \bigO{\alpha}
    \leq                           %
    \alphamore^3 \indegree{u} + 3 + \bigO{\alpha},
  \end{align*}
  as desired.
\end{proof}

\subsection{Update times}

It remains to establish the running times claimed in
\reftheorem{worst-case-mixed-apx}. By employing the same auxiliary
data structures described in \refsection{amortized-running-time}, we
may assume that each step in the pseudocode takes constant time. The
analysis is then largely reduced to counting the total number of arcs
processed by \algo{check-inc} and \algo{check-dec}.

We first re-establish the amortized bounds.

\begin{lemma}
  \labellemma{wc-update-amortized}
  \labellemma{wc-amortized}
  Each edge insertion and edge deletion takes $\bigO{\log{n} /
    \eps^2}$ amortized time.
\end{lemma}
\begin{proof}
  The argument is similar to that of the amortized data structure
  (\reflemma{amortized-running-time}, \refsection{amortized}), so we
  restrict ourselves to a sketch.

  We first discuss edge insertion. Here the running time is
  proportional to the number of arcs processed in the loop of
  \algo{check-inc}. The key points to amortizing the number of arcs is
  as follows. First, each edge insertion results in the in-degree of
  exactly one vertex increasing by $1$. Second, in order for an arc
  $a = (u,v)$ to be processed in the loop \algo{check-inc($u$)},
  $\indegree{v}$ must have increased from $\lab{v}{a}$ by a
  $\parof{1 + \bigOmega{\alpha}}$-factor. These factors allow us to
  apply the same charging scheme as in
  \reflemma{amortized-running-time} to obtain the
  $\bigO{\log{n} / \eps^2}$ amortized running time.

  Edge deletion is similar. The running time is proportional to the
  total number of arcs flipped over all recursive calls. An arc is
  flipped only if $\indegree{u}$ has decreased by a
  $\parof{1+ \bigOmega{\alpha}}$-factor since $\lab{u}{a}$ was last
  set. Additionally each edge deletion results in decreasing the
  in-degree of exactly one vertex, by $1$. These factors allow us to
  apply the same charging scheme as in
  \reflemma{amortized-running-time} to obtain the
  $\bigO{\log{n}/\eps^2}$ amortized running time.
\end{proof}

Now we analyze worst-case bounds.
\begin{lemma}
  \labellemma{wc-insert-time}
  Each edge insertion and edge deletion takes
  \begin{math}
    \bigO{\log{\arboricity{G} + \log{n} / \eps} / \alpha^2} = %
    \bigO{\log{\lpopt + \log{n} / \eps} \log{n}^2 / \eps^4}
  \end{math}
  worst-case time.
\end{lemma}
\begin{proof}
  Consider first edge insertion.  The running time is proportional to
  the number of arcs processed by \algo{check-inc}. Each call to
  \algo{check-inc} processes at most $\bigO{1/\alpha}$ arcs and makes
  at most one recursive call. Each recursive call
  \algo{check-inc($v$)} increases $\indegree{v}$ by a
  $\alphamore$-approximate factor over the previous call, and
  $\indegree{v}$ is bounded above by
  $\bigO{\arboricity{G} + \log{n} / \eps}$ by
  \reflemma{wc-local-optimality}. Thus there are at most
  $\bigO{\log[1 + \alpha]{\arboricity{G} + \log{n} / \eps}} =
  \bigO{\log{\arboricity{G} + \log{n} / \eps} / \alpha}$ recursive
  calls.

  The running time for edge deletion follows by analyzing
  \algo{check-dec} similarly.
\end{proof}

\subsection{Fractional orientations}
\labelsection{worst-case-duplicate-space}

As with the data structure in \refsection{amortized}, the data
structure here can be modified to implicitly simulate duplicate arcs
without increasing the space usage. The adjustments are the same as in
\refsection{amortized} so we limit ourselves to a sketch. The main
idea is to use the same labels $\lab{u}{a}$ and $\lab{v}{a}$ for all
copies of the same arc $a$. That is, when we update the labels of one
copy of $a$ it automatically propagates to all copies of $a$. As
before, this does not create any issues because labels are only made
when they are safe \wrt the inequalities in
\cref{lemma:wc-labels,lemma:wc-label-u,lemma:wc-label-v}, and if it is
safe to relabel one copy of an arc, it is safe to relabel all copies
of the arc. The other adjustments discussed in
\refsection{amortized-duplicate-space} extend here in a
straightforward fashion.

\section{Improved worst-case updates for small arboricity}

\labelsection{small-arboricity}

For densest subgraph, we now have the following worst case update
times:
\begin{enumerate}
\item $\bigO{\log{n} \log{\lpopt} / \eps^4}$ time for a
  $\parof{1 + \eps, \log{n} / \eps}$-bicritera approximation.
\item $\bigO{\log{n}^3 \parof{\log{\lpopt} + \log \log{n}} / \eps^6}$
  time for a $\parof{1 + \eps}$-factor approximation.
\end{enumerate}

Note that the first data structure has a faster update time but the
additive error implies that it is only good when density is
$\bigOmega{\log{n} / \eps^2}$.  The second running time is slower
because we implicitly duplicate edges to artificially increase the
arboricity to be at least $\bigOmega{\log{n} / \eps^2}$.  Now, in the
small-arboricity regime that necessitates the slower running time, the
$\log{\lpopt}$-factor is also negligible and the second running time
is closer to $\bigO{\log{n}^3 / \eps^6}$. So we have faster running
times in the ``high-aroboricity'' and ``low-arboricity'' regimes,
taken separately. The goal in this section is to unify these ideas and
obtain a faster running time over all.

To obtain a faster worst-case running time for densest subgraph, we
will build on ideas in the the previous section to develop a data
structure whose output is only valid in the low-arboricity
regime. This data structure will then be run in parallel with the
faster data structure mentioned above that does not duplicate edges,
which is both (a) accurate in the high-arboricity regime and (b)
correctly signals if we are in a high- or low-arboricity setting.
Overall, the data structure will still retain linear space and the
same amortized running times. The bounds we obtain for the
low-arboricity setting is as follows.

\begin{theorem}
  \labeltheorem{low-arboricity} Let $\eps, T > 0$ be given with $\eps$
  sufficiently small and $T > \bigOmega{\log{n} / \eps^2}$.  In
  $\bigO{\log{n} / \eps^2}$ amortized time and
  $\bigO{\log{n}^2 \log{T} / \eps^4}$ worst-case time per edge
  insertion or deletion, one can maintain an orientation of $G$, and
  (implicity) the vertices of a subgraph of $G$, such that if
  $\lpopt \leq T$, then the maximum in-degree is at most
  $\epsmore \lpopt + \bigO{\log{n} / \eps}$, and the density of the
  subgraph is at least $\epsless \lpopt - \bigO{\log{n} / \eps}$.
\end{theorem}

Before proving \reftheorem{low-arboricity}, we complete the discussion
of how to apply it to obtain faster worst-case updates for dynamic
densest subgraph in general. We first observe that, like the data
structures in the previous section, the additive
$\bigO{\log{n} / \eps}$-factor can be removed by implicitly
duplicating each edge $\bigO{\log{n} / \eps^2}$ times.  This increases
the running times by a $\bigO{\log{n} / \eps^2}$ factor but keeps
everything else (and in particular the space) the same. (To maintain
linear space usage, one makes the exact same adjustments as in
\refsection{worst-case-duplicate-space}.) Second, we set
$T = \bigO{\log{n} / \eps^2}$, since above this threshold we already
faster update times. Putting these ideas together we obtain the
following.

\begin{corollary}
  \labelcorollary{log-arboricity} Let $G$ be an undirected graph over
  $n$ vertices dynamically updated by edge insertions and
  deletions. Let $\eps > 0$ be sufficiently small. Let
  $T = \bigO{\log{n} / \eps^2}$. In $\bigO{\log{n}^2 / \eps^4}$
  amortized time and
  \begin{math}
    \bigO{\log{n}^3 \parof{\log \log{n} + \log{1/\eps}} / \eps^6}
  \end{math}
  worst-case time per edge insertion or deletion, one can maintain a
  fractional orientation with the following property.
  \begin{enumerate}
  \item If the maximum in-degree is at least $\epsmore T$, then
    $\lpopt \geq T$.
  \item If the maximum in-degree is at most $\epsmore T$, then the
    maximum in-degree is at most an $\epsmore \lpopt$. In this case,
    the data structure also implicitly provides a list-representation
    of the vertices of a $\epsless$-approximate densest subgraph.
  \end{enumerate}
\end{corollary}

By running the data structure in \refcorollary{log-arboricity} in
parallel with the data structure from
\reftheorem{worst-case-mixed-apx}, we obtain the following bounds for
dynamically approximating the densest subgraph.

\begin{corollary}
  \labelcorollary{wc-best}
  Let $G$ be an undirected graph over $n$ vertices dynamically updated
  by edge insertions and deletions. Let $\eps > 0$ be sufficiently
  small.  Then one can maintain an $\epsless$-approximation of the
  density, and an implicit list-representation of the vertices of an
  $\epsless$-approximate densest subgraph, with linear space and
  within the following time bounds.
  \begin{mathresults}
  \item $\bigO{\log{n} / \eps^2}$ amortized time per edge insertion or
    deletion.
  \item \begin{math}
      \bigO{\frac{\log{n} \log{\lpopt}}{\eps^4} +
        \frac{\log{n}^3 \parof{\log \log{n} + \log{1/\eps}}}{\eps^6}}
    \end{math}
    worst-case time per edge insertion or deletion.
    % \item \begin{math}
    %     \bigO{\frac{\log{n} + \log{\arboricity{G}}}{\eps^2} +
    %     \frac{\log{n}^2 \parof{\log \log{n} + \log{1/\eps}}}{\eps^4}
    %   }
    %   \end{math}
    %   worst-case time per edge deletion.
  \end{mathresults}
\end{corollary}

We now focus on proving \reftheorem{low-arboricity} for the remainder
of this section.

\subsection{High-level overview}

\begin{figure}[t]
  \begin{framed}
    \small
    \begin{algorithm}{check-inc}{$v$}
      \begin{blockcomment}
        We call this routine whenever $\indegreeT{v}$ has increased
        (always by $1$).
      \end{blockcomment}
    \item \labelstep{truncated-check-inc-loop} For up to $C/\alpha$ arcs
      $a = (u,v) \in \incut{v}$ s.t.\
      \begin{math}
        \indegreeT{v} \geq \parof{1 + \alpha / 2} \lab{v}{a},
      \end{math}
      in increasing order of $\lab{v}{a}$, for a sufficiently large
      constant $C$:
      \begin{steps}
      \item If
        $\indegreeT{v} \geq \alphamore \parof{\indegreeT{u} + 1}$:
        \begin{steps}
        \item \labelstep{truncated-check-inc-flip} Flip $a$ to $(v,u)$ and
          set $\lab{u}{a} = \indegreeT{u}$ and
          $\lab{v}{a} = \indegreeT{v}$.
          \begin{blockcomment}
            This restores $\indegreeT{v}$ to its previous value.
          \end{blockcomment}
        \item Call \algo{check-inc($u$)} and return.
        \end{steps}
      \item \labelstep{truncated-increase-label} Otherwise set
        $\lab{v}{a} = \indegreeT{v}$ and $\lab{u}{a} = \indegreeT{u}$.
      \end{steps}
      % \begin{blockcomment}
      %   If we reach this point outside the loop, then we were unable to
      %   decrease $\bary{v}$ by flipping arcs. For the amortized analysis:
      %   Spread $\bigO{1/\alpha}$ tokens uniformly over $\incut{v}$.
      % \end{blockcomment}
    \end{algorithm}

    \smallskip

    \begin{algorithm}{check-dec}{$u$}
      \begin{blockcomment}
        We call this routine whenever $\indegreeT{u}$ has decreased
        (always by $1$).
      \end{blockcomment}
    \item \labelstep{truncated-check-dec-loop} For up to $C/\alpha$ arcs
      $a = (u,v) \in \outcut{u}$ s.t.\
      \begin{math}
        \indegreeT{u} \geq \parof{1 + \alpha / 2} \lab{u}{a},
      \end{math}
      in increasing order of $\lab{u}{a}$, for a sufficiently large
      constant $C$: % \commentcode{Revised}
      \begin{steps}
      \item If
        $\indegreeT{v} \geq \alphamore \parof{\indegreeT{u} + 1}$:
        \begin{steps}
        \item \labelstep{truncated-check-dec-flip} Flip $a$ to $(v,u)$ and
          set $\lab{u}{a} = \indegreeT{u}$ and
          $\lab{v}{a} = \indegreeT{v}$.
          \begin{blockcomment}
            This restores $\indegreeT{v}$ to its previous value.
          \end{blockcomment}
        \item Call \algo{check-inc($u$)} and return.
        \end{steps}
      \item \labelstep{truncated-decrease-label} Otherwise set
        $\lab{v}{a} = \indegreeT{v}$ and $\lab{u}{a} = \indegreeT{u}$.
      % \item Let $a \in \outcut{u}$ maximize $\lab{u}{a}$. If
      %   \begin{math}
      %     \lab{u}{a} > \alphamore^2 \parof{\indegreeT{u} + 1}:
      %   \end{math}
      %   \begin{steps}
      %   \item Flip $a$ to $(v,u)$. Set $\lab{u}{a} = \indegreeT{u}$ and
      %     $\lab{v}{a} = \indegreeT{v}$.
      %     \begin{blockcomment}
      %       This restores $\indegreeT{u}$ to its previous value.
      %     \end{blockcomment}
      %   \item Call \algo{check-dec($v$)} and return.
      %   \end{steps}
      % \item For up to $C/\alpha$ arcs $a \in \incut{u}$ such that
      %   $\lab{u}{a} \geq \parof{1 + \alpha/4} \indegreeT{u}$, in
      %   decreasing order of $\lab{u}{a}$, for a sufficiently large
      %   constant $C$:
      %   \begin{steps}
      %   \item \labelstep{truncated-decrease-label-u} Set
      %     $\lab{u}{a} = \indegreeT{u}$.
      %   \end{steps}
      \end{steps}
    \end{algorithm}
  \end{framed}
  \vspace{-1em}
  \caption{Revised versions of \algo{check-inc} and \algo{check-dec}
    (cf.\ \reffigure{unweighted-worst-case-code}) that only maintains
    a locally optimal orientation whenever the arboricity is less than
    $T$.\labelfigure{truncated-code}\labelfigure{small-arboricity-code}}
\end{figure}

The new data structure takes the data structure from the previous
section and incorporates one simple idea: a threshold $T$.
For a vertex $v$, let
\begin{align*}
  \indegreeT{v} \defeq \min{\indegree{v}, T}.
\end{align*}
$\indegreeT{v}$ truncates the in-degree of $v$ to be at most $T$.

Rather than trying to minimize the maximum in-degree $\indegree{v}$,
the new data structure tries to minimize the maximum \emph{truncated}
in-degree $\indegreeT{v}$. Local optimality conditions are redefined
in terms of $\indegreeT{v}$, and we now only adjust arcs when we find
errors \wrt $\indegreeT{v}$. An important consequence is that
recursive calls effectively end at vertices with in-degree greater
than $T$. This effectively replaces ``$\log{\lpopt}$''-factors in
the running times from \refsection{worst-case} with
``$\log{T}$''-factors, even when $\arboricity{G}$ is much greater than
$T$.

We now describe the changes more precisely.  Let $T > 0$ be a fixed
value; $T$ should be at least $c_T \log{n} / \eps$ for a sufficiently
large constant $c_T > 0$. As before, let $\alpha = \eps^2 /
\log{n}$. The labels $\lab{u}{a}$ and $\lab{v}{a}$ of an arc $a$ are
now set to $\indegreeT{u}$ and $\indegreeT{v}$.  In
\algo{check-inc($v$)}, we only process an arc $a = (u,v)$ if
$\indegreeT{v}$ significantly exceeds $\lab{v}{a}$, and then we only
flip $a$ if $\indegreeT{v}$ significantly exceeds $\indegreeT{u}$.
(This is as opposed to acting on $\indegree{v}$ and $\indegree{u}$ in
the previous section.)  Likewise, in \algo{check-dec($u$)}, we only
process an arc $a = (u,v)$ if $\indegreeT{u}$ is significantly less
than $\lab{u}{a}$ and then we flip $a$ if $\indegreeT{u}$ is
significantly less than $\indegreeT{v}$. Note that the labels are
always at most $T$, and the data structure does not flip arcs when the
concerned vertex has in-degree much greater than $T$. In particular we
permit violations of the local optimality criteria when the in-degrees
exceed $v$.

Incorporating the threshold $T$ has the running time advantage of
ending the recursion calls when the in-degrees start to exceed
$T$. Consequently the recursive depth is reduced from
$\bigO{\log{\arboricity{G}} / \alpha}$ to $\bigO{\log{T} /
  \alpha}$. However it raises technical issues as well.  As mentioned
above, the data structure will certainly not maintain local optimality
for arcs when the in-degrees are larger than $T$. We need to redefine
a new notion of local optimality for the truncated in-degrees, and
show that they are sufficient for global optimality when
$\lpopt \leq T$.  It is also no longer clear how violations to large
in-degree might corrupt arcs where the in-degrees and labels are below
$T$. For example, it is not clear that a ``rogue'' arc violating local
optimality in the original sense will correct itself when the
in-degrees of its endpoints fall below $T$.

The pseudocode is presented in \reffigure{truncated-code}. It is
obtained by taking the pseudocode in the previous section and
replacing $\indegree{\cdots}$ with $\indegreeT{\cdots}$ everywhere.

\subsection{A truncated local optimality condition}

\labelsection{truncated-local-optimality}

We first address the issue of whether maintaining local optimality for
arcs below the threshold $T$ suffices to obtain global optimality, at
least in the restricted setting where $\lpopt \leq \apxless T$. (Note
that \emph{a priori} the data structure can have in-degrees greater
than $T$ even if $\lpopt \leq T$, as $\lpopt$ fluctuates above and
below $T$.)  The following lemma is similar to
\reflemma{apx-local=>global} except $\indegreeT{\cdots}$ takes the
role of $\indegree{\cdots}$.

\begin{lemma}
  Let $\alpha, \beta > 0$ with $\alpha = o(\log{n})$.  Let $T > 0$
  Suppose that for every arc $a = (u,v)$, we have
  \begin{math}
    \indegreeT{v} \leq \alphamore \indegreeT{u} + \bigO{1}.
  \end{math}
  Let $\mu = \max_v \indegreeT{v}$.  Then
  \begin{align*}
    \mu \leq e^{\bigO{\sqrt{\alpha \log{n}}}} \parof{\lpopt +
    \bigO{\sqrt{\frac{\log{n}}{\alpha}}} \beta}.
  \end{align*}
  In particular, for $\beta = \bigO{1}$, and
  $\alpha \leq c \eps^2 / \log{n}$ for a sufficiently small constant
  $c > 0$, we have
  \begin{align*}
    \mu \leq \epsmore \lpopt + \bigO{\log{n} / \eps}.
  \end{align*}
\end{lemma}

\begin{proof}[Proof sketch]
  The claim follows from the exact same proof as
  \reflemma{apx-local=>global}, except with $\mu$ now equal to
  $\max_v \indegreeT{v}$ instead of $\max_v \indegree{v}$.
\end{proof}

\subsection{Maintaining truncated local optimality}

We now show that the data structure maintains the truncated local
optimality conditions described in
\refsection{truncated-local-optimality}, for
$\alpha = \bigO{\eps^2 / \log{n}}$ and $\beta = \bigO{1}$. The
high-level structure is similar to the analysis for the previous data
structure in \refsection{wc-optimality}, and the details of the proofs
are similar as well.

\begin{lemma}
  \labellemma{sa-labels}
  For every arc $a = (u,v)$, we have
  \begin{align*}
    \lab{v}{a} \leq \alphamore \parof{\lab{u}{a} + 1}.
  \end{align*}
\end{lemma}

\begin{proof}[Proof sketch.]
  The proof ideas are essentially the same as \reflemma{wc-labels}. In
  short, $\lab{u}{a}$ and $\lab{v}{a}$ are set to $\indegreeT{u}$ and
  $\indegreeT{v}$ only in situations where the claimed inequality is
  satisfied.
\end{proof}

\begin{lemma}
  \labellemma{sa-label-v}
  For every arc $a = (u,v)$, we have
  \begin{align*}
    \indegreeT{v} \leq \alphamore \parof{\lab{v}{a} + 1}
  \end{align*}
\end{lemma}

\begin{proof}[Proof sketch.]
  The proof ideas are essentially the same as for
  \reflemma{wc-label-v}, and we restrict ourselves to a sketch.  We
  can define notions of $a$ being \emph{dangerous} and \emph{bad}
  based on when $\indegreeT{v}$ is large enough to being to threaten
  the desired inequality, and when $a$ actually violates the
  inequality, respectively. Similar to the proof
  \reflemma{wc-label-v}, we have the fact once $a$ becomes dangerous,
  $\indegreeT{v}$ still has to increase by a
  $\parof{1 + \bigOmega{\alpha}}$-factor for $a$ to be bad. We then
  argue that one of the calls to \algo{check-inc($v$)} from these
  increments would have to process $a$ before $\indegreeT{v}$ was
  large enough to make $a$ bad. Processing $a$ either resets the
  labels for $a$ or flips it; either way $a$ is no longer dangerous.
\end{proof}

\begin{lemma}
  \labellemma{sa-label-u}
  For every arc $a = (u,v)$, we have
  \begin{math}
    \lab{u}{a} \leq \alphamore \parof{\indegreeT{u} + 1}.
  \end{math}
\end{lemma}

We omit the proof of \reflemma{sa-label-u} as it is essentially the
same as the proof \reflemma{wc-label-u}, in the same way that the
proof of \reflemma{sa-label-v} matches the proof of
\reflemma{wc-label-v}.

\begin{lemma}
  \labellemma{sa-local-optimality} The data structure maintains
  $\parof{1+\bigO{\alpha}, \bigO{1}}$-local optimality.
\end{lemma}
\begin{proof}[Proof sketch]
  The claim follows from
  \cref{lemma:sa-labels,lemma:sa-label-u,lemma:sa-label-v} in the
  exact same way that \reflemma{wc-local-optimality} follows from
  \cref{lemma:wc-labels,lemma:wc-label-u,lemma:wc-label-v}.
\end{proof}

\subsection{Running time analysis}

\begin{lemma}
  Each edge insertion and deletion takes $\bigO{1 / \alpha} =
  \bigO{\log{n} / \eps^2}$ amortized time.
\end{lemma}

We refer the reader to the proof of \reflemma{wc-amortized}, which can
be applied here with essentially no changes except the argument is
now based on the truncated in-degrees.

\begin{lemma}
  Each edge insertion and deletion takes
  $\bigO{\log{n} \log{T} / \eps^4}$ worst-case time.
\end{lemma}

\begin{proof}[Proof sketch]
  The proof is similar to that of \reflemma{wc-insert-time} in
  \refsection{worst-case}. Consider first edge insertion. The two key
  ideas from that proof are as follows.  First, each call to
  \algo{check-inc($v$)} takes $\bigO{1/\alpha}$ time excluding
  recursive calls.  Second, each recursive call to
  \algo{check-inc($v$)}, $\indegree{v}$ increases by a
  $\alphamore$-factor. The difference now is that we recursion stops
  when $\indegree{v}$ exceeds $T$. The total running time is thus
  \begin{math}
    \bigO{\log[1+\alpha]{T} / \alpha} = \bigO{\log{T} / \alpha^2} =
    \bigO{\log{T} \log{n}^2 / \eps^4},
  \end{math}
  as desired.

  Likewise the key ideas used to bounding edge deletion in
  \reflemma{wc-insert-time} carry over here; the difference now is
  that the depth of the recursion is $\bigO{\log{T} / \alpha^2}$. The
  running time follows.
\end{proof}

\section{Extending to hypergraphs}

\labelsection{extensions}

\providecommand{\weight}{\fparnew{w}}%

% \subsection{Hypergraphs}

Recall that a hypergraph generalizes undirected graphs by allowing each
edge to have any number of endpoints. Let $G = (V,E)$ be a
hypergraph. The maximum number of endpoints in any edge is called the
\emph{rank} of the hypergraph; we let $\therank$ denote the rank of
$G$. The \emph{size} of the hypergraph is defined as the sum, over all
edges, of the number of endpoints in that edge. We let $\thesize$
denote the size of $G$.  For a set of vertices $S \subseteq V$, let
$E(S)$ denote the set of hyperedges with all endpoints in $S$. The
density of a set $S$ is defined as
\begin{math}
  \sizeof{E(S)} / \sizeof{S}.
\end{math}
The densest subhypergraph problem is to find the set $S\subseteq V$
that maximizes the density. This can be solved optimally via
a reduction to network flow or via submodular function minimization.

To generalize our data structures to hypergraphs we need to generalize
the notion of orientations to hypergraphs in a natural fashion. An orientation of a
hypergraph $G = (V,E)$ consists of selecting, for each edge $e \in E$,
an endpoint $v \in e$ called the \emph{head}. In this case we say that
$e$ is \emph{directed to $e$}. Given an orientation of $G$, the
in-degree of a vertex $v$, denoted $\indegree{v}$, is defined as the
number of edges for which $v$ is the head. One can define fractional
orientations of hypergraphs analogously; here each edge $e$ is
associated with a convex combination of endpoints that fractionally
act as the head of $e$.

In hypergraphs, as in graphs, it is easy to see that the density of
any subgraph is bounded above by the maximum in-degree of any
orientation.

Recall the dual LPs for densest subgraph and fractional orientations
in \reffigure{dsglp}. One can easily generalize these LPs to
hypergraphs, as noted in prior work. To extend the densest subgraph
LP, for each summand corresponding to an edge $e$ in the objective, we
take the minimum over all endpoints in $e$. In the dual LP for
fractional orientations, we now have a variable $y(e,v)$ for every
edge $e$ and every endpoint $v \in e$, and these must sum to at least
one for every edge $e$. For the remainder of this section, we let
$\lpopt$ denote the common optimum values of these LPs for the
hypergraph $G$.

Next, we define a notion of approximate local optimality of
orientations of hypergraphs that leads to approximate densest
subgraphs.  Let $G = (V,E)$ by a hypergraph and fix an orientation of
$G$. Let $\alpha, \beta \geq 0$. We say that the orientation is
$(\alpha, \beta)$-locally optimal, or a \emph{local
  $(\alpha,\beta)$-approximation}, if for all edges $e \in E$ with
head $v \in e$, and all other endpoints $u \in e$, we have
\begin{align*}
  \indegree{u} \leq \alphamore \indegree{v} + \bigO{1}.
\end{align*}

\begin{lemma}
  Let $\eps \in (0,1)$. Let $G = (V,E)$ be an oriented hypergraph and
  let $\mu = \max_v \indegree{v}$. Suppose the orientations is
  $\parof{c \eps^2 / \log{n}, \bigO{1}}$-locally optimal for a
  sufficiently small constant $c$. Then
  \begin{math}
    \mu \leq \epsmore \lpopt + \bigO{\ln{n} / \eps}.
  \end{math}
\end{lemma}

We omit the proof as it is essentially the same as
\reflemma{apx-local=>global}. Here we point out that in oriented
hypergraphs, for a set of vertices $S$, $\outneighbors{S}$ is defined
as the set of vertices that are an endpoint to an edge directed
towards a vertex in $S$.

The data structures for graphs generalize in a straightforward
fashion.  The main difference is that we have a label $\lab{v}{e}$ for
every edge $e$ and every endpoint $v$ of $e$. \algo{check-inc($v$)} is
generalized as follows. Recall that in graphs, \algo{check-inc($v$)}
processes an edge $e$ oriented to $v$ when $\indegree{v}$ has grown
significantly larger than $\lab{v}{e}$. In this case it either makes a
favorable flip to the smaller in-degree endpoint, or relabels the
endpoints of $e$. Both of these ideas generalize to hyperedges. The
difference is that we now have to check all the endpoints of $e$ to
identify the endpoint of minimum in-degree. If the minimum in-degree
is smaller (or substantially smaller) than that of the head, then we
make the minimum in-degree endpoint the head and recurse on that
vertex; otherwise we relabel all the endpoints of $e$. In particular,
processing an edge $e$ now takes time proportional to the number of
endpoints of $e$, which is at most the rank $\therank$ of the
hypergraph. Likewise \algo{check-dec($u$)} generalizes to hypergraphs
in a straightforward manner with an additional running time overhead
of $\therank$. The proof of correctness follows by the exact same
arguments as for graphs and is therefore omitted.

The following theorem extends \reftheorem{wc-mixed-apx} to hypergraphs.

\begin{theorem}
  \labeltheorem{wc-hypergraph-mixed-apx} Let $G$ be an unweighted and
  undirected hypergraph over $n$ vertices and rank $\therank$,
  dynamically updated by edge insertions and deletions.  The one can
  maintain an orientation of $G$ with maximum in-degree
  $\epsmore \lpopt + \bigO{\log{n} / \eps}$, and (implicitly) a
  subgraph of density $\epsless \lpopt - \bigO{\log{n}/\eps}$, in
  $\bigO{\therank \log{n} / \eps^2}$ amortized time and
  $\bigO{\therank \log{n}^3 \log{\lpopt + \log{n} / \eps} / \eps^4}$
  worst-case time per edge insertion or deletion.  The data structure
  uses $\bigO{\thesize}$ space.
\end{theorem}

We can also extend the truncated data structure from
\refsection{small-arboricity} to hypergraphs as described above, and
run it parallel with \reftheorem{wc-hypergraph-mixed-apx} similar to
the combination presented in \refsection{small-arboricity}.  Here,
however, a nonzero value of $\lpopt$ may be as small as $1/
\therank$. Thus to diminish the additive error in
\reftheorem{wc-hypergraph-mixed-apx} we have to duplicate each edge
$\bigO{\therank \ln{n} / \eps^2}$ time, rather than
$\bigO{\ln{n} / \eps^2}$ as in graphs.  Altogether we obtain the
following theorem generalizing \refcorollary{wc-best} to hypergraphs
of rank $\therank$.

\begin{theorem}
  Let $G$ be an undirected hypergraph of rank $\therank$ over $n$
  vertices dynamically updated by edge insertions and deletions. Let
  $\eps > 0$ be sufficiently small. Then one can maintain a
  $\epsless$-approximation of the density, and an implicit
  list-representation of the vertices of an $\epsless$-approximate
  densest subgraph, with linear space and within the following time
  bounds.
  \begin{mathresults}
  \item $\bigO{\therank^2 \log{n}^2 / \eps^4}$ amortized time per edge
    insertion or deletion.
  \item
    $\bigO{\frac{\therank \log{n} \log{\lpopt}}{\eps^4} +
      \frac{\therank^2 \log{n}^3 \parof{\log \log{n} +
          \log{1/\eps}}}{\eps^6}}$ worst-case time per edge insertion
    or deletion.
  \end{mathresults}
\end{theorem}

%\clearpage

\printbibliography

\appendix

\end{document}